\documentclass[11pt,letterpaper]{article}
\usepackage{natbib}

\usepackage[margin=1.1in]{geometry}
\usepackage[utf8]{inputenc}
\usepackage{amsfonts}
\usepackage{graphicx}
\usepackage{wrapfig}
\usepackage{pgfplots}
\usepackage{environ}
\pgfplotsset{compat=1.4}
\usepackage{amsmath}
\usepackage{amsthm}
\usepackage{amssymb}
\usepackage{amsfonts}
\usepackage{mathrsfs}
\usepackage{aliascnt}
\usepackage{natbib}
\usepackage{color}
\usepackage{tikz}
\usetikzlibrary{patterns}
\usepackage{subfig}
\usepackage{mdwlist}
\usepackage{bbm}
\usepackage{authblk}

\usepackage{cleveref}

\newtheorem{definition}{Definition}
\newtheorem{lemma}[definition]{Lemma}
\newtheorem{corollary}[definition]{Corollary}

\newtheorem{theorem}[definition]{Theorem}

\newtheorem{proposition}[definition]{Proposition}

\newcommand{\AutoAdjust}[3]{\mathchoice{ \left #1 #2  \right #3}{#1 #2 #3}{#1 #2 #3}{#1 #2 #3} }
\newcommand{\Xcomment}[1]{{}}

\newcommand{\InBrackets}[1]{\AutoAdjust{[}{#1}{]}}% {\left[{#1}\right]}
\newcommand{\Ex}[2][]{\operatorname{\mathbb E}_{#1}\InBrackets{#2}}

\def\expect{\Ex}

\newcommand{\OPT}{*}
\DeclareMathOperator*{\argmin}{arg\,min}

\newcommand{\agent}{i}
\newcommand{\agentiter}{\agent}
\newcommand{\numagents}{n}
\newcommand{\agentalt}{\agent^{\prime}}

\newcommand{\action}{a}
\newcommand{\actions}{\mathbf{a}}
\newcommand{\actionagent}{a_{\agentiter}}
\newcommand{\actionothers}{\actions_{-\agentiter}}
\newcommand{\actspace}{\mathbf A}
\newcommand{\actspaceagent}{A_{\agentiter}}

\newcommand{\actionalt}{a'}
\newcommand{\actionaltagent}{\actionalt_\agent}
\newcommand{\actionsitem}{\actions^\mitem}
\newcommand{\actionsalt}{\actions'}
\newcommand{\actionsaltothers}{\actionsalt_{-\agentiter}}

\newcommand{\actiondist}{G}

\newcommand{\actionsdist}{\mathbf{\actiondist}}
\newcommand{\actionsdistothers}{\actionsdist_{-\agentiter}}
\newcommand{\actionsdistitem}{\actionsdist^\mitem}

\newcommand{\allocation}{x}
\newcommand{\altallocation}{y}
\newcommand{\alloc}{\mathbf \allocation}
\newcommand{\allocopt}{\alloc^{\OPT}}
\newcommand{\allocalt}{\mathbf \altallocation}
\newcommand{\allocagent}[1][\agentiter]{\allocation_{#1}}
\newcommand{\allocoptagent}{\allocation^{\OPT}_{\agentiter}}
\newcommand{\allocdev}{z}
\newcommand{\agentallocalt}{\altallocation}
\newcommand{\allocaltagent}{\agentallocalt_\agent}
\newcommand{\allocaltagentitem}{\agentallocalt_{\agent}^{\mitem}}
\newcommand{\allocaltitem}{\mathbf{\agentallocalt}^{\mitem}}
\newcommand{\allocitem}{\alloc^\mitem}

\newcommand{\alloclevel}{\allocation^{\prime}}
\newcommand{\allocmin}{\underline\allocation_\agent}

\newcommand{\bidallocation}{\tilde\allocation}
\newcommand{\bidalloc}{\mathbf\bidallocation}

\newcommand{\bidallocagent}{\bidallocation_\agent}

\newcommand{\bidallocagentitem}{\bidallocagent^\mitem}

\newcommand{\bid}{a}
\newcommand{\bids}{\mathbf{\bid}}
\newcommand{\bidagent}{\bid_{\agent}}
\newcommand{\bidothers}{\bids_{-\agent}}

\newcommand{\deviation}{d}
\newcommand{\dev}{\deviation}
\newcommand{\mech}{M}

\newcommand{\mechanism}[1]{\textsc{#1}}
\newcommand{\mechanismres}[1]{\mechanism{#1}_{\reserves}}
\newcommand{\mechanismmon}[1]{\mechanism{#1}_{\reservesm}}

\newcommand{\FPA}{\mechanism{FPA}}
\newcommand{\FPAr}{\mechanismres{FPA}}
\newcommand{\FPAm}{\mechanismmon{FPA}}

\newcommand{\SPA}{\mechanism{SPA}}

\newcommand{\APA}{\mechanism{APA}}

\newcommand{\AP}{\mechanism{AP}}
\newcommand{\WPB}{\mechanism{WPB}}

\newcommand{\GFP}{\mechanism{Gfp}}

\newcommand{\OPTmech}{\mechanism{Opt}}
\newcommand{\OPTm}{\OPTmech}

\newcommand{\pay}{p}
\newcommand{\payment}{\mathbf{\pay}}
\newcommand{\paymentagent}{\pay_{\agentiter}}

\newcommand{\paymentitem}{\payment^\mitem}

\newcommand{\bidpay}{\tilde \pay}
\newcommand{\bidpayment}{\mathbf{\bidpay}}
\newcommand{\bidpaymentagent}{\bidpay_{\agentiter}}
\newcommand{\bidpaymentagentitem}{\bidpay_{\agent}^{\mitem}}

\newcommand{\WEL}{\textsc{Welfare}}
\newcommand{\REV}{\textsc{Rev}}
\newcommand{\UTIL}{\textsc{Util}}
\newcommand{\rev}{\REV}
\newcommand{\revpar}{\mu}

\newcommand{\strategy}{s}
\newcommand{\strat}{\mathbf{\strategy}}

\newcommand{\stratagent}{\strategy_\agent}
\newcommand{\stratothers}{\strat_{-\agent}}

\newcommand{\threshold}{\tau} %this is threshold
\newcommand{\thresholdagent}{\threshold_{\agentiter}}
\newcommand{\thresholdagentitem}{\thresholdagent^\mitem}

\newcommand{\thresholdbidagent}{\tilde t_{\agentiter}}

\newcommand{\expthreshold}{T}
\newcommand{\expthresholdagent}{\expthreshold_\agent}
\newcommand{\expthreshagentitem}{\expthresholdagent^\mitem}

\newcommand{\thresholdragent}[1]{\expthresholdagent^{#1}}
\newcommand{\expthresholdbidagent}{\tilde T_\agent}

\newcommand{\threshexpectedallpay}[2]{\expthresholdagent^{\AP}}					%expected threshold between alloc prob #1 and #2
\newcommand{\threshexpectedwinnerpaysbid}[2]{\expthresholdagent^{\WPB}}					%expected threshold between alloc prob #1 and #2

\newcommand{\threshpt}[1]{\thresholdagent(#1)}							%pointwise threshold
\newcommand{\threshexpected}[2]{\expthresholdagent}					%expected threshold between alloc prob #1 and #2
\newcommand{\threshexpgfp}[3]{\expthresholdagent^{#1}(#3)}

\newcommand{\threshbound}{\underline{\expthreshold}}						%threshold lowerbound
\newcommand{\threshboundagent}{\threshbound_{\agentiter}}				%threshold lowerbound for an agent
\newcommand{\threshboundexpected}[2]{\threshboundagent}			%expected threshold between alloc prob #1 and #2
\newcommand{\threshbidbound}{\underline{\expthresholdbidagent}}

\newcommand{\bidfi}{\effbidagent}

\newcommand{\effbid}{t}

\newcommand{\effbidagent}{\effbid_{\agent}}

\newcommand{\effbidallpay}{\effbidagent^{\AP}}
\newcommand{\effbidwinnerpaysbid}{\effbidagent^{\WPB}}

\newcommand{\equivbid}{\tilde{\beta}}
\newcommand{\equivbidagent}{\equivbid_\agent}

\newcommand{\equivbidagentitem}{\equivbidagent^\mitem}

\newcommand{\equivbidact}[1]{\equivbidagent(#1)}

\newcommand{\effaction}{\alpha_\agent}

\newcommand{\ppu}{\effbidagent}

\newcommand{\util}{u}
\newcommand{\utilagent}{\util_{\agentiter}}

\newcommand{\bidutil}{\tilde u}
\newcommand{\bidutilagent}{\bidutil_{\agentiter}}

\newcommand{\val}{v}
\newcommand{\vals}{\mathbf \val}
\newcommand{\valothers}{\vals_{-\agentiter}}
\newcommand{\valagent}{v_{\agentiter}}

\newcommand{\valueset}{V}
\newcommand{\values}{\mathbf\valueset}
\newcommand{\valuesagent}{\valueset_\agent}

\newcommand{\valuecdf}{F_{\agent}}
\newcommand{\valuepdf}{f_{\agent}}

\newcommand{\dist}{F}

\newcommand{\valuecdfs}{\mathbf{\dist}}
\newcommand{\valuecdfsothers}{\valuecdfs_{-\agent}}

\newcommand{\vval}{\phi}
\newcommand{\vvalagent}{\vval_{\agentiter}}

\newcommand{\REVpos}{\REV^{+}}

\newcommand{\mitem}{j}
\newcommand{\mechitem}{\mech_\mitem}

\newcommand{\actionagentitem}{\actionagent^\mitem}

\newcommand{\wants}{\mathcal S}
\newcommand{\wantsagent}[1][\agent]{\wants_{#1}}

\newcommand{\numitems}{m}

\newcommand{\reserve}{r}
\newcommand{\reserves}{\mathbf{\reserve}}
\newcommand{\reserveagent}{\reserve_{\agent}}

\newcommand{\reserveagentm}{\reserveagent^{\OPT}}
\newcommand{\reservesm}{\mathbf{\reserve}^{\OPT}}
\newcommand{\reservesothers}{\reserves_{-\agentiter}}

\newcommand{\group}{B}

\newcommand{\NOTEC}[1]{#1}
\newcommand{\EC}[1]{}
\newcommand{\IFECELSE}[2]{#2}

\newcommand{\restate}[2]{\vspace{0.1in} \noindent{\bf #1 (Restatement).}
{\em #2} }

\begin{document}

% Page heads
\markboth{J Hartline et al.}{Price of Anarchy for Auction Revenue}

% Title portion
\title{Price of Anarchy for Auction Revenue\footnote{ We thank Vasilis
    Syrgkanis for comments on a prior version of this paper for which
    simultaneous composition did not hold, for suggesting the study of
    simultaneous composition, and for perspective on price-of-anarchy
    methodology.} } \author[1]{Jason Hartline} \author[ ]{Darrell Hoy}
\author[2]{Sam Taggart}

\affil[1]{Northwestern University}
\affil[2]{Oberlin College}
%\author{JASON HARTLINE
%\affil{Northwestern University}
%DARRELL HOY
%\affil{Northwestern University}
%SAM TAGGART
%\affil{Northwestern University}}

\maketitle

\begin{abstract}
This paper develops tools for welfare and revenue analyses of
Bayes-Nash equilibria in asymmetric auctions with single-dimensional
agents.  We employ these tools to derive approximation results for
social welfare and revenue.  Our approach separates the smoothness
framework of, e.g., \citet{ST13}, into two distinct parts, isolating
the analysis common to any auction from the analysis specific to a
given auction.  The first part relates a bidder's contribution to
welfare in equilibrium to their contribution to welfare in the optimal
auction using the price the bidder faces for additional
allocation. Intuitively, either an agent's utility and hence
contribution to welfare is high, or the price she has to pay for
additional allocation is high relative to her value. We call this
condition \emph{value covering}; it holds in every Bayes-Nash
equilibrium of any auction. The second part, \emph{revenue covering},
uses the auction rules and feasibility constraints to relate the
revenue of the auction to the prices bidders face for additional
allocation.  Combining the two parts gives approximation results to
the optimal welfare, and, under the right conditions, the optimal
revenue.

As a centerpiece result, we analyze the single-item first-price
auction with individual monopoly reserves (the price that a monopolist
would post to sell to that agent alone; these reserves are generally
distinct for agents with values drawn from distinct distributions).
When each distribution satisfies the regularity condition
of \citet{M81}, the auction's revenue is at least a $2e/(e-1) \approx
3.16$ approximation to the revenue of the optimal auction. We also
give bounds for matroid auctions with winner-pays-bid or all-pay
semantics, the generalized winner-pays-bid position auction, and
winner-pays-bid single-minded combinatorial auctions.  Finally, we
give an extension theorem for simultaneous composition, i.e., when
multiple auctions are run simultaneously, with single-valued,
unit-demand agents.

\end{abstract}

\newpage

%\category{J.4}{Social and Behavioral Sciences}{Economics}

%\terms{Mechanism Design, Bayesian Mechanism Design, Approximation}

%\keywords{Mechanism Design, Bayesian Mechanism Design, Approximation, Simple Mechanisms}

%\acmformat{Jason Hartline, Darrell Hoy, and Sam Taggart, 2014. Revenue Covering: Approximately Optimal Revenue and Welfare In Bayes-Nash Equilibrium.}

\section{Introduction}

%
% benefit of price of anarchy analysis
%
The first step of a classical microeconomic analysis is to solve for
equilibrium.  Consequently, such analysis is restricted to settings
for which equilibrium is analytically tractable; these settings are
often disappointingly idealistic.  Worst-case equilibrium analysis
(often referred to as ``price of anarchy analysis'') provides an
alternative approach. Instead of solving for equilibrium, properties
of equilibrium can be quantified from consequences of best response.
These methods have been primarily employed for analyzing social
welfare.  While welfare is a fundamental economic objective, there are
many other properties of economic systems that are important to
understand.  This paper gives methods for deriving worst-case bounds
for auction revenue.

%
% beyond welfare.
%
Equilibrium requires that each agent's strategy be a best response to
the strategies of others.  A typical worst-case approximation analysis obtains
a bound on the social welfare (the sum of the revenue and all agent
utilities) from a lower bound on an agent's utility implied by best
response.  Notice that the agents themselves are each directly
attempting to optimize a term in the objective.  This property makes
social welfare special among objectives.  Can simple best-response
arguments be used to quantify and compare other objectives?  This
paper considers the objective of revenue, i.e., the sum of the agent
payments.  Notice that each agent's payment appears negatively in her
utility and, therefore, she prefers smaller payments; collectively the
agents prefer smaller revenue.

%
% mechanism design vs PoA 
%
The agenda of this paper parallels a recent trend in mechanism
design.  Mechanism design typically seeks to identify a mechanism with
optimal performance in equilibrium.  Optimal mechanisms tend to be
complicated and impractical; consequently, a recent branch of
mechanism design has sought to quantify the loss between simple
mechanisms and optimal mechanisms.  These simple (designed) mechanisms
have carefully constructed equilibria (typically, the truthtelling
equilibrium).  The restriction to truthtelling equilibrium, though
convenient in theory, is problematic in practice \citep[][]{AM06}.
In particular, this truthtelling equilibrium is specific to an ideal agent
model and tends to be especially non-robust to out-of-model phenomena.
The worst-case equilibrium analysis instead considers the
performance of simple mechanisms absent a carefully constructed
equilibrium.

%
% example: single-item auction
%
As an example, consider the single-item first-price auction, in which 
agents place
sealed bids, the auctioneer selects the highest bidder to win, and the
winner pays her bid.  The fundamental tradeoff faced by the agents in
selecting a bidding strategy is that higher bids correspond to a higher
probability of winning (which is beneficial) but higher payments (which is detrimental).
This first-price auction is a fundamental auction in practice
and it is the role of auction theory to understand its performance.
When the agents' values for the item are drawn independently and
identically, first-price equilibria are well-behaved: the 
symmetry of the setting enables the easy computation
of equilibrium \citep{Krishna09}, the equilibrium is unique
\citep{CH13,L06,MR03}, and the highest-valued agent always wins (hence,
social welfare is maximized).  When the agents' values are
non-identically distributed, analytically solving for equilibrium is
notoriously difficult.  For example, \citet{V61} posed the question
of solving for equilibrium with two agents with values drawn uniformly
from distinct intervals; this problem was finally resolved half a
century later by \citet{KZ12}.

%equilibrium being analytically tractable vs. reality of reaching equilibrium
%The intractibility of solving analytically for equilibrium is foremost a problem of theory. It does not rule out BNE as a practical concept: agents can reach equilibrium by playing learning strategies, numerically solving the differential equations  implied by equilibrium, etc. Free from the demands of theoretical analysis, agents may use these heuristic techniques, may focus on specific instances of their optimization problem, and may employ algorithmic techniques such as those developed by \citet{JL10}, rather than pursuing a general, analytical characterization.

Worst-case approximation analysis allows us to make general statements about 
behavior in equilibrium without requiring an explicit characterization. For example, 
\citet{ST13} show that the first-price auction's social welfare in
equilibrium is at least an $e/(e-1) \approx 1.58$ approximation to the optimal
social welfare, and moreover, this bound continues to hold if multiple
items are sold simultaneously by independent first-price auctions.  Importantly, this  analysis sidesteps the
intractability of solving for equilibrium and instead derives its
bounds from simple best-response arguments.

%Fundamentally, if a precondition for understanding the first-price
%auction in equilibrium is solving for its equilibrium, then this
%understanding is unlikely to be

%In terms of analysis methods, it is striking to compare the
%price-of-anarchy analyses \citep[e.g.,][]{ST13}, for auctions (as games
%of incomplete information) to the analyses of classical auction
%theory \citep[e.g.,][]{M81}.  The price of anarchy analyses, as
%described above, employ simple best-response arguments.  Analyses for
%complex equilibrium concepts such as Bayes-Nash equilibrium are given
%indirectly by extension theorems applied to analyses of simple
%equilibrium concepts such as pure Nash equilibrium.  On the other
%hand, a classical auction theoretic analyses begin from the
%characterization of Bayes-Nash equilibrium which gives (a) a
%monotonicity condition on the equilibrium allocation in terms of an
%agents preference (intuitively higher values should lead to no smaller
%allocation probabilities) and (b) a payment identity which expresses
%the equilibrium payment made in terms of the allocation.

\subsection{Approach}
Our analysis comprises two main arguments. The first, \emph{value
  covering}, encapsulates the welfare consequences of best
response. For any given agent, either their expected utility is high,
or they are unable to get allocated cheaply because their threshold
prices are high. Specifically, we show that for any agent in BNE, the
sum of their utilty and expected threshold bid is at least a
$\tfrac{e-1}{e}$-fraction of their value and, therefore, of their
contribution to the expected welfare of any mechanism. The second
argument, \emph{revenue covering}, captures the mechanism-specific
details that affect equilibrium welfare. An auction is revenue covered
if whenever allocation is difficult to achieve (i.e.\ threshold bids
are high), this difficulty translates into revenue for the
auctioneer. 
%In particular, any agent's expected threshold bid must be
%upper bounded by the revenue of the mechanism. 
Combining these two properties yields welfare results akin to those
proven in \citet{ST13}.

By decomposing welfare analysis into this modular framework, we are
able to extend the arguments to revenue analysis. The Bayes-Nash
equilibrium characterization of \citet{M81} reduces revenue
optimization to welfare maximization in the space of ``virtual
values.'' We adopt a similar approach: using the value covering
argument from our welfare analysis, we show that an agent's positive
virtual surplus plus their expected threshold bid approximates their
virtual value. If equilibrium revenue is not too badly reduced by
negative virtual surplus, this aforementioned \emph{virtual value
  covering} argument, combined with revenue covering, implies that the
first price auction's revenue approximates that of the optimal
mechanism. To control the negative virtual surplus from the
first-price auction, we use two common methods. First, we extend value
and revenue covering to mechanisms with reserves. This analysis shows
that adding individual monopoly reserves to the first-price auction
prevents the allocation of agents with negative virtual values, which implies a revenue bound. The second method we use begins with an
extension of a theorem of \citet{BK96}. They argue that with
sufficient competition in a second-price auction with i.i.d.\ agents,
the revenue loss from negative virtual surplus is small. We extend
this argument to asymmetric first-price and all-pay auctions. As a
corollary, we obtain a revenue bound for first-price auctions with
sufficient competition.

We present the above approach in
\Cref{sec:easy}, \Cref{sec:revenue},
and \Cref{sec:beyondsingleitem} in the context of auctions with
winner-pays-bid semantics. In \Cref{sec:beyondpyb}, we show how
to extend the above approach to mechanisms with other payment
semantics. We show that value covering is a property of BNE of
single-dimensional mechanisms in general. Consequently, all that
is required to show approximation results is revenue covering for the
mechanism being considered. As an example, we show that the all-pay
auction is revenue covered, yielding that in equilibrium the welfare
of the all-pay auction approximates the optimal welfare, and that
with sufficient competition, the equilibrium revenue approximates the
optimal revenue as well. Finally, we show that revenue covering is
robust to the simultaneous composition of mechanisms. Consequently,
our welfare (and some revenue results) extend to simultaneously run
mechansisms as well.

\subsection{Results}
%do we need this?
For single-item and matroid auctions (where the feasibility constraint
is given by a matroid set system), we give welfare and revenue results
with both winner-pays-bid and all-pay payment formats.  The
winner-pays-bid variants of these auctions (a) solicit bids, (b)
choose an outcome to optimize the sum of the reported bids of served
agents, and (c) charge the agents that are served their bids. The
all-pay variants of these auctions (a) solicit bids, (b) choose an
outcome to optimize the sum of the reported bids of served agents, and
(c) charge all agents their bids. Our analyses of these auctions are
compatible with reserve prices.

\emph{Welfare.} In first-price auctions and winner-pays-bid matroid
auctions, we derive a welfare approximation bound of $e/(e-1)$.  These
results also extend to the generalized first-price position
auction. For all-pay auctions in the above environments, we use the
same proof framework to show a welfare bound of $2e/(e-1)$. The
all-pay result is not the best-known bound, but we show how to modify
our methods to obtain the best-known welfare bound of $2$. While these
results do not improve on the best-known bounds, our proofs are
notable in that they can be extended to revenue.

\emph{Revenue.} For winner-pays-bid single-item and matroid auctions
with monopoly reserves and regular distributions, we show that the
equilibrium revenue is at most a factor of $2e/(e-1)$ from optimal. The
same bound holds in the generalized first-price position auction with
monopoly reserves.  If instead of reserves each bidder must compete
with at least one duplicate bidder, the approximation bound for revenue in
first-price auctions is at most $3e/(e-1)$; in all-pay auctions, at
most $6$.

\emph{Simultaneous Composition.} We also show via an extension theorem that the above welfare bounds (and revenue results for winner-pays-bid auctions with reserves) hold when auctions are run simultaneously if agents are \emph{unit-demand} and \emph{single-valued} across the outcomes of the auctions.

%
% the second part
%
%% We show that a simple property of mechanisms, \emph{revenue covering}
%% is a sufficient condition for approximately optimal welfare in all
%% BNE, and when combined with a method of limiting.

%% For the first-price auction, this gives a Price of Anarchy bound of
%% $\frac{e}{(e-1)}$. When combined with monopoly-reserves in a regular,
%% matroid environment, the revenue is always at least a
%% $\frac{2e}{(e-1)}$-approximation to the optimal auction. When at least
%% two agents from every type must compete over allocation, the revenue
%% is always at least a $\frac{3e}{(e-1)}$-approximation to the optimal
%% auction. These are the first known revenue bounds for the first-price
%% auction in the asymmetric setting.

%% We show that this behavior also holds for all-pay auctions; for
%% generalized first price position auctions, and when many auctions are
%% run simultaneously, assuming single-valued and unit demand bidders.

\subsection{Related Work}

%smoothness
Understanding welfare in games without solving for equilibrium is a
central theme in the smooth games framework of \citet{R09} and the
smooth mechanisms extension of \citet{ST13}. A core principle of smoothness is that restricting the arguments used in proving the smoothness property dictate how broadly the result extends. One way to view our work is that we limit our proofs a  way that allows for extensions to revenue approximations. 
 
Our framework refines the smoothness framework for Bayesian games in
three notable ways. First, we decompose smoothness into two
components, separating the the consequences of best-response (value
covering) from the specifics of a mechanism (revenue
covering). Second, because we focus on the optimization problem that
individual bidders are facing, we can attain results that only hold
for certain bidders --- for instance, bidders with values above their
reserve prices. Third, we only consider the Bayesian setting, which
allows us to employ the BNE characterization of \citet{M81} to
approximate revenue and to relate other formats of auctions to
winner-pays-bid formats via a framework of equivalent bids.

We note two subsequent works with strong connections to our
decomposition of smoothness into value covering and revenue
covering. First, \citet{DK15} show that revenue covering, which they
call ``permeability,'' is in fact a necessary condition (we show
sufficiency) for the equilibrium welfare of a mechanism to be proven
to be good via the smoothness framework.  Second, \citet{HNS15} show
how to derive empirical welfare bounds by measuring the degree to
which value and revenue covering hold, rather than inferring agents'
true values.

%simple vs optimal
A number of papers have derived revenue guarantees for
the welfare-optimal Vickrey-Clarke-Groves (VCG) mechanism in asymmetric settings. \citet{HR09}
show that VCG with monopoly reserves, a carefully chosen anonymous reserve, or duplicate bidders
achieves revenue that is a constant approximation to the revenue
optimal auction. \citet{DRY10} show that the single-sample
mechanism, essentially VCG using a single bid as a reserve, 
achieves approximately optimal
revenue in broader settings. \citet{RTY12} showed that in broader
environments, including matching settings, limiting the supply of
items in relation to the number of bidders gives a constant
approximation to the optimal auction.

%economics - kirkegaard
In the economics literature, \citet{K09} shows that
understanding the ratios of expected payoffs in equilibria of asymmetric auctions can lead to insights into equilibrium structure. \citet{K12a} considers properties of distributions on which the revenue of the first price auction exceeds that of the second price auction, and vice versa. \citet{L06} and \citet{MR03} establish
equilibrium uniqueness in the asymmetric setting with some assumptions
on the distributions of agents.
%
%
%\subsection{Layout of Paper}
%
%We warm up with the single-item first-price auction with monopoly reserves in Section~\ref{sec:easy}. We then develop the revenue covering 
%framework and show welfare approximation results in Section~\ref{sec:revenuecovering}. In Section~\ref{sec:revenue}, we give revenue approximation results for revenue covered mechanisms. In Section~\ref{sec:welfare}, we show mechanisms in which revenue covering holds. In Section~\ref{sec:simul} we show the extension theorem that revenue covering holds under
%simultaneous composition.

\section{Preliminaries}
\label{sec:prelim}

\paragraph{Bayesian Mechanisms.}
This paper considers mechanisms for $n$ single-dimensional agents with
linear utilities. Each agent has a private value for service,
$\valagent$, drawn independently from a distribution $\valuecdf$ over
valuation space $\valuesagent$. We write $\valuecdfs=\prod_\agent
\valuecdf$ and $\values=\prod_\agent \valuesagent$ to denote the joint
value distribution and space of valuation profiles, respectively,
where $\actspaceagent$ is the set of possible actions for $\agent$.  A
\emph{mechanism} consists of an action space $\actspace = \prod_\agent
\actspaceagent$, a bid allocation rule $\bidalloc$, and a bid payment
rule $\bidpayment$, mapping actions of agents to probabilistic
allocations and payments respectively. Each agent $\agent$ draws their
private value $\valagent$ from $\valuecdf$ and selects an action
according to some strategy $\stratagent: \valuesagent \rightarrow
\actspaceagent$.  We write $\strat=(\strategy_1,\ldots,\strategy_n)$
to denote the vector of agents' strategies. Given the actions
$\actions=(\action_1,\ldots,\action_n)$ selected by each agent, the
mechanism computes $\bidalloc(\actions)\in [0,1]^n$ and
$\bidpayment(\actions)$. Each agent's utility is
$\bidutilagent(\actions)=\valagent\bidallocagent(\actions)-\bidpaymentagent(\actions)$.

Typically mechanisms operate with constraints on permissible
allocations. Examples include single-item environments,
$\bidalloc(\actions)$ must satisfy $\sum_i \bidallocagent(\actions)
\leq 1$ for all actions $\actions \in \actspace$, or matroid
environments, where $\bidalloc(\actions)$ must be the membership
vector for an underlying matroid (see
Section~\ref{sec:greedy}). We will denote the set of feasible
allocations for a given environment by $\mathcal F$.

%Mechanisms typically operate with constraints on permissible allocations. A \emph{feasibility environment} specifies the set of feasible allocation vectors. Mechanisms for a feasibility environment must choose only allocations from the feasible set. The simplest example is a single-item auction, in which at most one person at a time can be served. %This paper assumes feasibility environments are \emph{downward-closed}: if $(\allocation_1,\ldots,\allocation_k,\ldots,\allocation_n)$ is feasible, so is $(\allocation_1,\ldots,0,\ldots,\allocation_n)$. We will often consider the special case of \emph{matroid} environments, in which the set of feasible allocations correspond to the independent sets of a matroid set system.

Given a strategy profile $\strat$, we often consider the expected allocation and payment an agent faces from choosing some action $\actionagent\in\actspaceagent$, with expectation taken with respect to other agents' values and actions induced by $\strat$. We  treat $\strat$ as implicit and write $\bidallocagent(\actionagent)=\expect[\valothers]{\bidallocagent(\actionagent,\stratothers(\valothers))}$, with $\bidpaymentagent(\actionagent)$ and $\bidutilagent(\actionagent)$ defined analogously. 
Given $\strat$ implicitly, we also consider values as inducing payments and allocations. We write $\alloc(\vals)=\bidalloc(\strat(\vals))$ and $\payment(\vals)=\bidpayment(\strat(\vals))$, respectively. Furthermore, we will denote agent $\agent$'s interim allocation probability and payment by $\allocagent(\valagent)=\bidallocagent(\stratagent(\valagent))$ and $\paymentagent(\valagent)=\bidpaymentagent(\stratagent(\valagent))$. We define $\util(\vals)$ and $\utilagent(\valagent)$ similarly. In general, we use a tilde to denote outcomes induced by actions, and omit the tilde when indicating outcomes induced by values. We refer to $\bidalloc$ as the \emph{bid allocation rule}, to distinguish it from $\alloc$, the \emph{allocation rule}. We adopt a similar convention with other notation.

\emph{Bayes-Nash Equilibrium.} A strategy profile $\strat$ is in {\em Bayes-Nash equilibrium} (BNE) if for all agents $\agent$, $\stratagent(\valagent)$ maximizes $\agent$'s interim utility, taken in expectation with respect to other agents' value distributions $\valuecdfsothers$ and their actions induced by $\strat$.
That is, for all $\agent$, $\valagent$, and alternative actions $\actionalt$:
$\expect[\valothers]{\bidutilagent(\strat(\vals))}
 \geq 
\expect[\valothers]{\bidutilagent(\actionalt,\stratothers(\valothers))}.$

%We will consider only mechanisms where agents can gain from participation, regardless of their value --- that is, mechanisms that are \emph{interim individually rational}. We will thus assume that in any auction every bidder has at least one withdraw action $\withdrawaction$ such that $\bidallocagent(\withdrawaction,\actionothers)=0$ and $\bidpaymentagent(\withdrawaction,\actionothers)=0$ for any $\actionothers$. We assume all mechanisms have at least one such action for each agent. In any BNE, each agent has the option to withdraw and must therefore get non-negative utility.

\citet{M81} characterizes the interim allocation
and payment rules that arise in BNE. These results are summarized in the following theorem.
\begin{theorem}[\citealp{M81}]
\label{thm:myerson}
For any mechanism with $\paymentagent(0)=0$ and any value distribution $\valuecdfs$,  BNE implies the following:
\begin{enumerate*}
\item (monotonicity) 
\label{thmpart:monotone}
The interim allocation rule $\allocagent(\valagent)$ for each agent is monotone
  non-decreasing in $\valagent$.
\item 
\label{thmpart:payment}
(payment identity) The interim payment rule satisfies $\paymentagent(\valagent) = \valagent \allocagent(\valagent) - \int_0^{\valagent}\allocagent(z) d z$.
\item
\label{thmpart:revequiv}
(revenue equivalence)
Mechanisms and equilibria which result in the same interim allocation rule $\alloc$ must therefore have the same interim payments as well.
\end{enumerate*}
\end{theorem} 

\paragraph{Mechanism Design Objectives.}
We consider the problem of maximizing two primary objectives in expectation at BNE: utilitarian welfare and revenue. The revenue of a mechanism $\mech$ is the total payment of all agents. Given a mechanism $\mech$ and a distribution over action profiles $\actionsdist$, the revenue of $\mech$ under $\actionsdist$ is given by $\rev(\mech,\actionsdist)=\Ex[\actions\sim\actionsdist]{\sum_\agent \bidpaymentagent(\actions)}$. Alternatively, a strategy profile $\strat$ and value distribution $\valuecdfs$ jointly determine a distribution over action profiles. We may therefore also write the revenue of a mechanism $\mech$ under $\strat$ and  $\valuecdfs$ as $\rev(\mech,\strat,\valuecdfs)=\Ex[\vals\sim\valuecdfs]{\sum_\agent \bidpaymentagent(\strat(\vals))}$. The welfare
of a mechanism $\mech$ under a strategy profile $\strat$ and value distribution $\valuecdfs$ is the total utility of all participants
including the auctioneer; denoted $\WEL(\mech,\strat,\valuecdfs)
= \rev(\mech,\strat,\valuecdfs) + \Ex[\vals\sim\valuecdfs]{\sum_\agent \bidutilagent(\strat(\vals))}=\Ex[\vals\sim\valuecdfs]{\sum_\agent\valagent\bidallocagent(\strat(\vals))}$ We will also refer to welfare as ``surplus.'' For both welfare and revenue, we will suppress $\actionsdist$, $\strat$, and $\valuecdfs$ when context makes the distributions of bids and values clear.

Our welfare benchmark is the outcome that always serves the highest valued feasible agents. That is, we seek to approximate $\WEL(\OPTmech)=\expect[\vals]{\max_{\allocopt\in \mathcal F} \sum_\agent \valagent\allocoptagent}$. This can be implemented via the Vickrey-Clarke-Groves (VCG) mechanism. We measure a mechanism $\mech$'s welfare performance by its worst-case approximation ratio, given by
\begin{equation*}
\max_{\valuecdfs \in \text{Indep};\, \valuecdfs,\strat\in \text{BNE}(\mech,\valuecdfs)} \frac{\WEL(\OPTmech,\valuecdfs)}{\WEL(\mech,\strat,\valuecdfs)},
\end{equation*}
where BNE$(\mech,\valuecdfs)$ is the set of BNE for $\mech$ under value distribution $\valuecdfs$. 

For revenue, we will make extensive use of the characterization of revenue in \citet{M81} that follows from Theorem~\ref{thm:myerson}:
\begin{lemma}
\label{lem:vvals}
In any BNE $\strat$ for distributions $\valuecdfs$, the ex ante expected payment of an agent is
  $\expect[\valagent]{\paymentagent(\valagent)} = \expect[\valagent]{\vvalagent(\valagent)\allocagent(\valagent)}$, where
  $\vvalagent(\valagent) = \valagent - \frac{1-\valuecdf(\valagent)}{\valuepdf(\valagent)}$ is the {\em virtual value} for value $\valagent$. It follows that $\rev(\mech) = E_\vals[\sum_\agent\paymentagent(\vals)]=\expect[\vals]{\sum_\agent\vvalagent(\valagent)\allocagent(\vals)}$.
\end{lemma}
Using this result, \citet{M81} derives the revenue-optimal mechanism
for any value distribution $\valuecdfs$. This mechanism is
parameterized by the value distribution $\valuecdfs$, and the
optimality is in expectation over $\vals\sim\valuecdfs$. We
specifically consider distributions with no point masses where
$\vvalagent(\valagent)$ is monotone in $\valagent$ for each
$\agent$. Such distributions are said to be \emph{regular}.  If each
agent has a regular distribution, then the revenue-optimal mechanism
selects the allocation which maximizes
$\sum_\agent\vvalagent(\valagent)\allocagent(\vals)$.  For revenue, we
will again measure the performance of a mechanism by its worst-case
approximation ratio:
\begin{equation*}
\max_{\valuecdfs\in\text{Reg};\,\strat\in \text{BNE}(\mech,\valuecdfs)} \frac{\rev(\OPTmech_\valuecdfs,\valuecdfs)}{\rev(\mech,\strat,\valuecdfs)},
\end{equation*} 
where $\text{Reg}$ is the set of regular distributions and $\OPTmech_\valuecdfs$ is the Bayesian revenue-optimal mechanism for value distribution $\valuecdfs$.

\section{Single-Item First Price Auction}\label{sec:easy}

We motivate our framework by analyzing the welfare of the first-price
auction, showing that it always approximates the welfare of the
welfare-optimal mechanism. This result has been known since the work
of \citet{ST13}, but our proof will lend itself to extension and
generalization. The rest of the paper will use this proof structure as
a template.

\begin{theorem}\label{thm:fpawelfare}
The welfare in any BNE of the first price auction is at least an $\frac{e}{e-1}$-approximation to the
welfare of the welfare-optimal mechanism. 
\end{theorem}

Our proof proceeds in two steps. First, we analyze the interim
optimization problem faced by every bidder. We quantify an intuitively
obvious tradeoff: either the bidder can get allocated cheaply,
attaining high utility, or allocation is expensive for that bidder to
obtain. Second, we note that allocation is only expensive to obtain if
the mechanism's revenue is high. This yields a tradeoff between
revenue (seller welfare) and utilities (buyer welfare):

\begin{equation} \label{eq:fpaoverall}
\sum\nolimits_\agent \UTIL_\agent(\FPA) + \REV(\FPA)  \geq \tfrac{e-1}{e} \WEL(\OPTm).
\end{equation}
This equation directly implies the theorem.

%TODO : move this somewhere else
%Note that equation~\eqref{eq:fpaoverall} resembles closely the
%inequality in the smooth games and mechanism frameworks \citep{ST13,
%  R09}. It differs primarily in that we are not defining a specific
%deviation, but rather deriving bounds explicitly from BNE.

\paragraph{A Bidder's Optimization Problem:} 
We now develop ideas needed to make this analysis formal. Consider the
problem faced by a bidder $\agent$ with value $\valagent$ in the first
price auction. Her expected utility for a possible bid $\dev$ is
$\bidutilagent(\dev)=(\valagent-\dev)\bidallocagent(\dev)$, where
$\bidallocagent(\dev)$ is the interim bid allocation rule she faces in
BNE. If we plot the bid allocation rule $\bidallocagent(\dev)$ for any
alternate bid $\dev$, then $\bidutilagent(\dev)$ is precisely the area
of the rectangle in the lower right; see
\Cref{fig:fpautilagent}. Let $\bidagent$ be her best response
bid given her value $\valagent$. It must be that $\bidagent$ maximizes
$\bidutilagent(\dev)$ and therefore the area of the rectangle under
$\bidallocagent(\dev)$.

	\begin{figure}[t]
	\centering
	\begin{tikzpicture}[xscale=3.5, yscale=3.5, domain=0:0.9, smooth]
		\draw  (1.12,1) -- (.91,0.7) .. controls ( 0.68,0.4) and ( 0.6,0.32) .. (0.16,0.1) --(0,0);
	    \node at (1.2,0.9) {$\bidallocagent(\dev)$};
	    %axes
	    \draw[-] (0,0) -- (0,1) node[left] {$1$};
		\draw[-] (0,0) -- (1.4,0) node[below] {Bid ($\dev$)};
		\node at (1, -0.09) {$\valagent$};
		\draw[pattern=north east lines, pattern color=lightgray] ( 0.5,0) rectangle ( 1,0.29);
		\node [fill=white!50] at (0.75,0.15){$\bidutilagent(\dev)$};
		\node at (-0.15,0.3) {$\bidallocagent(\dev)$};
		\draw[] (-0.02, 0.3) -- (0, 0.3);		
		\node at (0.5,-0.09) {$\dev$};
		\draw[] (1, 0) -- (1, -0.025);
		\draw[] (0.5, 0) -- (0.5, -0.025);
		\node at (0.7,1.1) {Bid Allocation Rule};
	\end{tikzpicture}
\caption{For any bid $\dev$ (and implicit value $\valagent$), the
  expected utility $\bidutilagent(\dev)$ is the area of a rectangle
  between $(\dev, \bidallocagent(\dev))$ on the bid allocation rule
  and $(\valagent, 0)$. The best-response bid $\bidagent$ is chosen to
  maximize this area.\label{fig:fpautilagent}}
	\end{figure}

When other bidders have realized values and submitted bids, bidder
$\agent$ wins only if her bid exceeds the bids of other
players. Consequently the price a bidder must pay to win is
$\threshpt{\valothers}= \max_{j\neq i} s_j (\val_j)$; we will refer to
it as her \emph{threshold bid}. In a Bayesian setting, a bidder reacts
not to a deterministic threshold, but rather views
$\threshpt{\valothers}$ as a random variable, with the bid allocation
rule $\bidallocagent(\cdot)$ as its CDF.

The expected threshold bid $\expthresholdagent=\Ex[\valothers]{\threshpt{\valothers}}$ of an agent is a rough measure of how hard it is for agent $\agent $ to receive allocation, and will be the focal quantity of our analysis. The expected value of a nonnegative random variable is the area above its CDF - in other words, agent $\agent$'s expected threshold bid is the area above $\bidallocagent(\cdot)$, which is $\int_{0}^{1} 1-\bidallocagent(z)\ dz$. It will be convenient to compute this quantity by integrating the inverse of $\bidallocagent(\cdot)$, which is given by $\bidfi(\allocation)= \min \{\bid\ |\ \bidallocagent(\bid)\geq
\allocation \}$. The inverse allocation function $\bidfi(\allocation)$ is the amount agent $i$ must bid to ensure allocation probability $\allocation$. In terms of $\bidfi(\cdot)$, the expected
threshold bid is $\expthresholdagent =
\int_{0}^{1} \bidfi(z)\ dz$.

%The geometry of agent $\agent$'s optimization problem is shown in Figure~\ref{fig:threshold}. 
%\footnote{This manner fo computing the expected threshold is more complicated than necessary for single-item auctions, but will extend cleanly to more general environments.}

%Let $\bidcdfothers(\bid)$ be the cumulative
%distribution function of the highest bids from other bidder. Then
%$\bidfi(\allocation)$ is either the reserve $\reservei$ or the bid
%required to beat the highest bid from other agents a $\allocation$
%fraction of the time, $\bidfi(\allocation) = \max(\reservei,
%\bidcdfothers^{-1}(\allocation))$.\footnote{If $\bidcdfothers$ is %not
%  invertible, then define $\bidcdfothers^{-1}(\allocation)$ to be %the
%  function $\bidcdfothers^{-1}(\allocation) = \inf\{
%  b\ |\ \bidallocagent(b)\geq \allocation\}$.}

%Let $\allocmin$ lowerbound feasible allocation above a bidders reserve. If there is no reserve, $\allocmin=0$; if there is a reserve $\reservei$, then $\allocmin=\bidallocagent(\reservei)$. In particular, we will focus on how expensive feasible allocation is, that is allocation starting with what a bidder can get with their reserve bid.%%

%TODO DH - change ui(bi) to ui(vi), change some bi to bi(vi)
%FPA intro pics: bid rectangle and upper bound
	\begin{figure}[t]
			\small
			\centering
			%\hspace*{\fill}
			\subfloat[][ As $\bidagent$ is a best-response to the actions of other agents, the indifference curve $\bidutilagent(\bidagent)/(\valagent-\dev)$ upper bounds $\bidallocagent(\dev)$. \label{fig:indifference}]{
	\begin{tikzpicture}[xscale=3.5, yscale=3.5, domain=0:0.9, smooth]
		%\draw[ pattern=north west lines, pattern color=lightgray] (0,0) -- (0,0.7) -- (.91,0.7) .. controls ( 0.68,0.4) and ( 0.6,0.32) .. (0,0);
		%\node [fill=white] at ( 0.3,0.45) {$\threshupto{\allocoptagent(\valagent)}$};
		\draw  (1.12,1) -- (.91,0.7) .. controls ( 0.68,0.4) and ( 0.6,0.32) .. (0.16,0.1) --(0,0);;
		\draw[line width=1.4pt] (0.85, 1) -- (0.82, 0.9) -- (.75,0.7) .. controls ( 0.55,0.19) and ( 0.35,0.20) .. (0,0.15);	 
		\node at (0.53,0.55) {\footnotesize $ \frac{\bidutilagent(\bidagent)}{\valagent-\dev}$};   
		\draw[pattern=north east lines, pattern color=lightgray] ( 0.5,0) rectangle ( 1,0.29);
	    \node at (1.2,0.9) {\footnotesize $\bidallocagent(\dev)$};
	    %axes
	    \node [fill=white!50] at (0.75,0.15){\footnotesize $\bidutilagent(\bidagent)$};
	    \draw[-] (0,0) -- (0,1) node[left] {$1$};
		\draw[-] (0,0) -- (1.4,0) node[below] {Bid ($\dev$)};
		\node at (1, -0.09) {\footnotesize $\valagent$};
		\node at (-0.16,0.3) {\footnotesize $\bidallocagent(\bidagent)$};		
		\draw[] (-0.02, 0.3) -- (0, 0.3);		
		%\node at (-0.12,0.7) {$\allocoptagent(\valagent)$};			
		\node at (0.5,-0.09) {\footnotesize $\bidagent$};
		\draw[] (1, 0) -- (1, -0.025);
		\draw[] (0.5, 0) -- (0.5, -0.025);
		\node at (-0.15, 0.15) {\footnotesize $\frac{ \bidutilagent(\bidagent)}{\valagent}$};
		\draw[] (-0.02, 0.155) -- (0, 0.155);
		%\node at (-.1, 1.08) {Allocation};
		\node at (1.1,1.1) {Bid Allocation Rule};
	\end{tikzpicture}} \IFECELSE{\hspace*{0.2cm}}{\hspace*{1cm}}
	\subfloat[][The expected threshold $\expthresholdagent$ is the area above the allocation rule. \label{fig:threshold} ]{
	\begin{tikzpicture}[xscale=3.5, yscale=3.5, domain=0:0.9, smooth]
		\draw  (1.12,1) -- (.91,0.7) .. controls ( 0.68,0.4) and ( 0.6,0.32) .. (0.16,0.1) --(0,0);;
	    \node at (1.2,0.9) {\footnotesize $\bidallocagent(\dev)$};
	    \draw[-] (0,0) -- (0,1) node[left] {$1$};
		\draw[-] (0,0) -- (1.4,0) node[below] {Bid ($\dev$)};
		\node at (1, -0.09) {\footnotesize $\valagent$};
		\draw[pattern=north east lines, pattern color=lightgray] ( 0.5,0) rectangle ( 1,0.29);
		\node [fill=white!50] at (0.75,0.15){\footnotesize $\bidutilagent(\bidagent)$};
		\draw[ pattern=north west lines, pattern color=lightgray] (0,0) -- (0,1) -- (1.12, 1) -- (.91,0.7) .. controls ( 0.68,0.4) and ( 0.6,0.32) .. (0.16,0.1) --(0,0);
		\node [fill=white] at ( 0.35,0.54) {\footnotesize $ \expthresholdagent$};
		\node at (-0.15,0.3) {\footnotesize $\bidallocagent(\bidagent)$};
		\draw[] (-0.02, 0.29) -- (0, 0.29);		
		\node at (0.5,-0.09) {\footnotesize $\bidagent$};
		\draw[] (1, 0) -- (1, -0.025);
		\draw[] (0.5, 0) -- (0.5, -0.025);
		\node at (0.7,1.1) {Bid Allocation Rule};
	\end{tikzpicture}
	}	
	\caption{ \label{fig:introfpa}
	}
\end{figure}

\paragraph{Relating Contributions to First-Price and Optimal Welfare:}  We will now approximate each bidder's
contribution to the optimal welfare individually, using the bidder's utility
in the first-price auction and a
fraction of the revenue in the first-price auction. In these terms, the steps to prove Theorem~\ref{thm:fpawelfare} are: 
\begin{enumerate*}
\item \emph{Value Covering}: Each bidder's utility and expected threshold together approximate her value. (Lemma \ref{lem:covering}) %$\valagent\allocagent(\valagent) + \threshexpected{\allocagent(\valagent)}{\allocoptagent(\valagent)} \geq  \frac{e-1}{e} \valagent\allocoptagent(\valagent)$  (Lemma \ref{lem:covering})
\item \emph{Revenue Covering}: The revenue of the first price auction upperbounds the expected thresholds of all agents. (Lemma \ref{lem:monoprevapx})
\end{enumerate*}

\noindent

\begin{lemma}[Value Covering]\label{lem:covering}
In any BNE of the first-price auction, for any bidder $i$ with value $\valagent$, 
%\allocoptagent(\valagent)\REV(\FPA)
\begin{equation}
\utilagent(\valagent)+\expthresholdagent \geq \tfrac{e-1}{e} \valagent
\label{eq:monvalcovering}.
\end{equation}
\end{lemma}

%Comparison of welfare and revenue approximation pictures.
	\begin{figure}[t]
			\small
			\centering
			%\hspace*{\fill}
			\subfloat[][Lemma \ref{lem:covering} shows the shaded areas cover a $(e-1)/e$ fraction of the dashed box, bidder $i$'s value $\valagent$ and hence maximum contribution to the optimal welfare. \label{fig:welfareboundgeo} ]{
	\begin{tikzpicture}[xscale=3.5, yscale=3.5, domain=0:0.9, smooth]
		\draw[ pattern=north west lines, pattern color=lightgray] (0,0) -- (0,1) -- (1.12, 1) -- (.91,0.7) .. controls ( 0.68,0.4) and ( 0.6,0.32) .. (0.16,0.1) --(0,0);
		\node [fill=white] at ( 0.39,0.56) {\footnotesize $\threshexpected{\allocagent(\valagent)}{\allocoptagent(\valagent)}$};
		\draw  (1.12,1) -- (.91,0.7) .. controls ( 0.68,0.4) and ( 0.6,0.32) .. (0.16,0.1) --(0,0);
	    \node at (1.2,0.9) {\footnotesize $\bidallocagent(\dev)$};
	    %axes
	    \draw[-] (0,0) -- (0,1) node[left] {$1$};
		\draw[-] (0,0) -- (1.4,0) node[below] {Bid ($\dev$)};
		\node at (1, -0.09) {\footnotesize $\valagent$};
		\draw[] (1, 0) -- (1, -0.025);
		\draw[pattern=north east lines, pattern color=lightgray] ( 0.5,0) rectangle ( 1,0.29);
		\node [fill=white!50] at (0.75,0.15){\footnotesize $\bidutilagent(\bidagent)$};
		%\draw[ pattern=north west lines, pattern color=lightgray] (0,0) -- (0,0.7) -- (.91,0.7) .. controls ( 0.68,0.4) and ( 0.6,0.32) .. (0,0);
		%\node [fill=white] at ( 0.3,0.45) {$\thresholdagent(\allocalt,\strat)$};
		\node at (-0.16,0.3) {\footnotesize $\allocagent(\valagent)$};	
		\draw[] (-0.02, 0.29) -- (0, 0.29);	
		%\node at (-0.12, 0) {$\allocmin$};			
		\node at (0.5,-0.09) {\footnotesize $\bidagent$};
		\draw[] (0.5, 0) -- (0.5, -0.025);
		\draw[line width=1.3pt, dashed] (0,0) rectangle (1,1);
		%\node at (-.1, 1.08) {Allocation};
		\node at (0.7,1.1) {Bid Allocation Rule};
	\end{tikzpicture}
	%REVENUE PICTURE - VIRTUAL VALUES
		}\IFECELSE{\hspace*{0.2cm}}{\hspace*{1cm}}
			\subfloat[][Lemma \ref{lem:vvalcoveringeasy} shows the shaded areas cover an $(e-1)/e$ fraction of $i$'s virtual value $\vvalagent(\valagent)$. \label{fig:revenueboundgeo}]{
	\begin{tikzpicture}[xscale=3.5, yscale=3.5, domain=0:0.9, smooth]
		\draw[ pattern=north west lines, pattern color=lightgray] (0,0) -- (0,1) -- (1.12, 1) -- (.91,0.7) .. controls ( 0.68,0.4) and ( 0.6,0.32) .. (0.16,0.1) --(0,0);
		\node [fill=white] at ( 0.39,0.56) {\footnotesize $\threshexpected{\allocagent(\valagent)}{\allocoptagent(\valagent)}$};
		\draw  (1.12,1) -- (.91,0.7) .. controls ( 0.68,0.4) and ( 0.6,0.32) .. (0.16,0.1) --(0.0,0);
	    \node at (1.2,0.9) {\footnotesize $\bidallocagent(\dev)$};
	    %axes
	    \draw[-] (0,0) -- (0,1) node[left] {$1$};
		\draw[-] (0,0) -- (1.4,0) node[below] {Bid ($\dev$)};
		\draw[] (1, 0) -- (1, -0.025);		
		\node at (1, -0.09) {\footnotesize $\valagent$};
		\draw[] (0.75, 0) -- (0.75, -0.025);		
		\node at (0.75, -0.09) {\footnotesize $\vvalagent(\valagent)$};
		\draw[pattern=north east lines, pattern color=lightgray] 
										( 0,0) rectangle ( 0.75, 0.29);
		\node [fill=white!50] at (0.35,0.15){\footnotesize $\vvalagent(\valagent)\allocagent(\valagent)$};
		\node at (-0.16,0.3) {\footnotesize $\allocagent(\valagent)$};	
		\draw[] (-0.02, 0.29) -- (0, 0.29);	
	
		\node at (0.5,-0.09) {\footnotesize $\bidagent$};
		\draw[] (0.5, 0) -- (0.5, -0.025);
		%\node at (-0.12, 0) {$\allocmin$};
		\draw[line width=1.5pt, dotted] (0,0) rectangle (0.75,1);
		%\node at (-.1, 1.08) {Allocation};
		\node at (0.9,1.1) {Bid Allocation Rule};
	\end{tikzpicture}
	}	
	\caption{ \label{fig:covering}}
\end{figure}
\noindent

\begin{proof} 
We will prove value covering in two steps: first, by developing a lower bound $\threshbound$ on the expected threshold $\expthresholdagent$; second, by optimizing to get the worst such bound. The first-price bid deviation approach of \citet{ST13} gives the same result.

%\emph{Lowerbounding $\expthreshold$.} 
%TODO ST: make not bad!!! 
 \textbf{Lowerbounding $\expthresholdagent$.} In best responding, bidder $\agent$ chooses an action which maximizes her utility. Hence for an agent with value $\valagent$ and any bid $\dev$, her equilibrium utility $\utilagent$ satisfies $\utilagent \geq (\valagent-\dev) \bidallocagent(\dev)$. We may write the righthand side in terms of the inverse allocation function $\bidfi(\cdot)$ to get $\utilagent \geq (\valagent-\bidfi(\allocation)) \allocation$. Rearranging this inequality yields a bound on $\bidfi(\allocation)$ for any $\allocation$: 
 \begin{equation}
 \bidfi(\allocation)\geq \valagent-\tfrac{\utilagent}{\allocation}.
 \notag
 \end{equation}
 
 Note that this bound is meaningful as long as the righthand side is nonnegative - that is, as long as $\allocation\geq \utilagent/\valagent$. Otherwise, note that $\bidfi(\allocation)\geq 0$. To derive a lower bound on the expected threshold $\expthresholdagent$, we use the definition of $\expthresholdagent$ and integrate $\bidfi(\allocation)$ over all relevant values of $\allocation$: from $\utilagent/\valagent$ to 1. Hence:
\begin{equation}
\expthresholdagent\geq\int_{\tfrac{\utilagent}{\valagent}}^1 \valagent-\tfrac{\utilagent}{\allocation}\,d\allocation=\threshboundagent.
\notag
\end{equation}

%In terms of $\ppu(\allocdev)$, we have $\utilagent(\bidagent) \geq (\valagent-\bidfi(z)) z$ and $\valagent-\frac{\utilagent}{z} \leq  \bidfi(z)$. %This bound is meaningful as long as $\valagent-\frac{\utilagent(\valagent)}{\allocdev}\geq 0$, or alternatively $\allocdev\geq\utilagent(\valagent)/\valagent$. It follows that

 \textbf{Worst-case $\threshboundagent$.}
%\emph{Optimizing $\threshboundagent$.} 
Evaluating the integral for $\threshboundagent$ gives $\threshboundagent =\valagent-\utilagent(1-\ln\frac{\utilagent}{\valagent})$, hence $\utilagent+\threshboundagent =\valagent+\utilagent\ln\frac{\utilagent}{\valagent}$. Holding $\valagent$ fixed and minimizing with respect to $\utilagent$ yields a minimum at $\utilagent=\valagent/e$, hence
$\utilagent + \threshboundagent \geq \tfrac{e-1}{e}\valagent$. The lemma follows.
\end{proof}
Note that this analysis depended only on the fact that bidder $\agent$ was best responding to a bid distribution. We later will derive a nearly identical condition for every single-dimensional mechanism in BNE using this same logic.

We now show that expected thresholds lowerbound revenue, which will combine with value covering to produce the welfare result. While value covering depended
only on equilibrium bidding behavior,
revenue covering will only depend on the form of the first price
auction, and will thus hold for arbitrary (not necessarily BNE)
bidding strategies.

%TODO - add qualification with reserves.
\begin{lemma}[Revenue Covering\footnote{To prove our welfare result, it suffices to show $\REV(\FPA)\geq \expthresholdagent$ for every agent $\agent$. We use the more complicated statement to parallel the statement more general feasibility environments.}]\label{lem:monoprevapx}
Fix an arbitrary bid distribution $\actionsdist$ for the first price auction. For any feasible allocation $\allocalt$, 
\begin{equation}
\label{eq:FPARC}
\REV(\FPA,\actionsdist) \geq  \sum_{\agent}\expthresholdagent\allocaltagent.
\end{equation}
\end{lemma}

%TODO - fix this proof!
\begin{proof} 
The revenue of a first price auction is the expected highest bid,
and $\expthresholdagent$ is the expected highest bid from all agents except $i$. Hence for any agent $\agent$, $\REV(\FPA,\actionsdist)\geq \expthresholdagent$. Since single-item feasible allocations sum to at most 1, equation (\ref{eq:FPARC}) follows.
\end{proof}

%\begin{proof}
%The expected threshold bid up to an allocation $\threshupto{\allocoptagent(\valagent)}$ can be viewed as the smallest $\allocoptagent(\valagent)$ fraction of highest bids from other bidders. In particular, this means $\threshupto{\allocoptagent(\valagent)}$ is smaller than any other way of distributing the bids with probability $\allocoptagent(\valagent)$. One such method of distribution is to give each bidder the winning payment when she would have won in the optimal allocation. Hence, 
%\begin{equation*}
%\threshupto{\allocoptagent(\valagent)} \leq \expect[\valothers]{ \allocoptagent(\vals) \left(\sum\nolimits_j \allocation_j(\vals) \bid_j(\vals)\right)}.
%\end{equation*}
%
%
%Summing and taking expectation over all agents and values gives
%
%\begin{align}
%\sum\nolimits_i \expect[\valagent]{\threshupto{\allocoptagent(\valagent)}} &\leq  \expect[\vals]{ \sum\nolimits_i \allocoptagent(\vals) \left(\sum\nolimits_j \allocation_j(\vals) \bid_j(\vals)\right)}\nonumber\\
%&=\expect[\vals]{\sum\nolimits_j \allocation_j(\vals) \bid_j(\vals)}\nonumber\\
%&=\REV(\FPA)
%\end{align}
%\
%\end{proof}
%Rev(FPA) > area above the allocation curve
%By the monotonicity of $\bid(z)$, $\int_0^{\allocoptagent(\valagent)} \bid(z) \ dz \leq \allocoptagent(\valagent) \REV(\FPA)$. Note that  $\int_0^{\allocoptagent(\valagent)} \bid(z) \ dz$ is precisely the area above the bid allocation rule up to $\allocoptagent(\valagent)$ in Figure~\ref{fig:fpautilagent}.

We now combine value and revenue covering to approximate the optimal welfare. 

\begin{proof}[\NOTEC{Proof }of \Cref{thm:fpawelfare}]
We begin by summing the value covering in equality for each agent in
an arbitrary value profile $\vals$:
\begin{equation}
\sum\nolimits_{\agent}\utilagent(\valagent)+\sum\nolimits_{\agent}\expthresholdagent\geq\frac{e-1}{e}\sum\nolimits_{\agent}\valagent.
\notag
\end{equation}
Let $\allocopt(\vals)$ be the allocation of the optimal mechanism for $\vals$. Since $\allocoptagent(\vals)\leq 1$ for each agent $\agent$, and since $\utilagent(\valagent)\geq 0$, we obtain:
\begin{equation}
\sum\nolimits_{\agent}\utilagent(\valagent)+\sum\nolimits_{\agent}\expthresholdagent\allocoptagent(\vals)\geq\frac{e-1}{e}\sum\nolimits_{\agent}\valagent\allocoptagent(\vals).
\notag
\end{equation}

Applying revenue covering and taking expectation with respect to $\vals$ shows that $\UTIL(\FPA) +
\REV(\FPA) \geq \frac{e-1}{e} \WEL(\OPTm)$. Since welfare is the sum of agent utilities and revenue, the welfare of the first price auction is an $e/(e-1)$ approximation to
$\OPTm$.
\end{proof}

\subsection{Welfare Lower Bounds}
The approximation results we have given in this section for the
single-item first-price auction are not known to be tight. In
Appendix~\ref{sec:app-example}, we present the best-known lower bound,
with an approximation factor of $1.15$. Note the large gap between
this lower bound and the upper bound of $\frac{e}{e-1} \approx 1.58$
from Theorem~\ref{thm:fpawelfare} and \citet{ST13}.

Beyond a single auction, \citet{CKST13} have shown that the
$\frac{e}{e-1}$ bound is tight for the simultaneous composition of
item auctions when bidders have submodular valuations.

%For revenue, the approximation ratio can be at least as bad as $2$, using the same lower bound \citet{HR09} show for VCG with monopoly reserves. The example has two bidders, one with deterministic value $1$, the other with value drawn according to the equal revenue distribution with support over $[1, H]$ for some large $H$. With a light perturbation of the distribution the monopoly price for the second bidder is $1$. Assuming ties go to bidder $2$, an equilibrium exists where both players bid $1$, giving revenue of $1$. The optimal auction however can set a reserve of $H$ for the second bidder and sell to the first bidder at price $1$ if the reserve is not met, achieving a revenue of $2$ as $H$ grows.

\section{Revenue Approximation}\label{sec:revenue}
Our welfare result hinges on the complementary relationship between
the utility of a bidder and the bids of other bidders in the
mechanism. Using this relationship to directly bound revenue is not as
straightforward. The results of \citet{M81}, however, provide another
method of accounting for each bidder's impact on revenue, their
\emph{virtual value}. Using virtual surplus in place of utilities
allows us to adapt our method for proving welfare guarantees to the
objective of revenue.

\subsection{Revenue}
The welfare of a mechanism can be expressed as the expected total
value of agents who are served. \citet{M81} demonstrated a similar
characterization of revenue in terms of the expected total virtual
value, reducing the problem of revenue maximiation to welfare
maximization. We adopt a similar approach, analyzing revenue using
tools developed for welfare.

We will begin by showing the analogue of value covering, \emph{virtual
  value covering}, in which each bidder's positive contribution to
equilibrium virtual welfare and expected threshold bid combine to
approximate her virtual value, which by \citet{M81} is upperbounds her
contribution to the optimal revenue.

\begin{lemma}[Virtual Value Covering] \label{lem:vvalcoveringeasy}
In any BNE of the first price auction, for any bidder $i$ with value $\valagent$ such that $\vvalagent(\valagent)\geq 0$, 
\begin{equation}
\vvalagent(\valagent)\allocagent(\valagent) + \threshexpected{\allocagent(\valagent)}{\alloclevel} \geq \tfrac{e-1}{e} \vvalagent(\valagent). \label{eq:vvalcoveringeasy}
\end{equation}
\end{lemma}  
\begin{proof}
First note that surplus is an upper bound on utility, i.e.: $\valagent\allocagent(\valagent)\geq\utilagent(\valagent)$. Combined with \Cref{lem:covering}, this implies that 
\begin{equation}
\valagent\allocagent(\valagent) + \threshexpected{\allocagent(\valagent)}{}\geq \tfrac{e-1}{e}\valagent.
\label{eq:valvalcovering}
\end{equation}
By the definition of virtual value as $\vvalagent(\valagent)=\valagent-\frac{1-\valuecdf(\valagent)}{\valuepdf(\valagent)}$, we have that $\vvalagent(\valagent)\leq\valagent$. Substituting $\vvalagent(\valagent)$ for $\valagent$ in \eqref{eq:valvalcovering} therefore only weakens the inequality, which implies the result.
\end{proof}

See \Cref{fig:revenueboundgeo} for an illustration. Intuitively,
value covering captures the idea that the expected threshold makes up
the difference between an agent's utility and their value. The
difference between virtual surplus and virtual value is proportionally
smaller, so the expected threshold can cover that gap as well.

\Cref{lem:vvalcoveringeasy} does not immediately imply a revenue
approximation result, as virtual value covering only implies an
approximation for agents with positive virtual value. A revenue
approximation result requires the revenue impact of agents with
negative virtual value to be mitigated as well. In
\Cref{sec:reserves}, we show that reserve prices suffice for this
purpose. In \Cref{sec:duplicates}, we prove that sufficient
competition also reduces negative virtual surplus and implies an
approximation result.

\subsection{Reserve Prices}
\label{sec:reserves}
In this section, we show how to adapt the framework of value and
revenue covering to accomodate auctions with reserves. The framework
was driven by two key arguments. Value covering showed that either an
agent was receiving high utility or faced large impediments to
obtaining allocation. Revenue covering captured the argument that when
the agent could not get allocated easily, the mechanism must be
obtaining high revenue. Reserve prices complicate this second
argument: an agent with a high reserve might face difficulty winning
because of their reserve, which, unlike other agents' bids, does not
generate revenue for the mechanism (if the agent loses). We solve this
problem below by proving relaxed versions of virtual value and revenue
covering. When combined with monopoly reserves
(i.e.\ $\reserveagent=\vvalagent^{-1}(0)$) for each agent, they will
combine to produce a revenue approximation result, albeit with a
slightly larger constant.

%We consider first price auctions with a reserve set to exclude exactly the agents with negative virtual values. Under the assumption that bidders' value distributions are regular, it suffices to set monopoly reserves at $\reserveagent=\vvalagent^{-1}(0)$. This reserve price acts as an additional impediment to allocation and contributes to agents' thresholds just like other agents' bids. Unlike bids, the reserve contributes no revenue when it prevents an agent from getting allocated. In other words, reserves may cause revenue covering not to hold.

%To salvage revenue covering, we note that the reserve (rather than bids) only binds to exclude an agent when all other agents bid below this reserve. The rest of the time, it is bids which determine allocation, and these bids translate to revenue. Intuitively, this suggests that a form of revenue covering holds much of the time.

The relaxed version of the value and revenue covering framework will
use as its pivotal quantity a restricted version of an agent's
expected threshold bid. As before, $\bidfi(\allocation)=\min
\{\bid\ |\ \bidallocagent(\bid)\geq \allocation \}$ denotes the
smallest bid that achieves allocation of at least $\allocation$. Now,
however, note that below $\bidallocagent(\reserveagent)$, it no longer
corresponds to the inverse cumulative distribution function of the
highest bids from all other agents. For $x\leq
\bidallocagent(\reserveagent)$, $\bidfi(\allocation)=\reserveagent$ -
the threshold comes from the reserve price. Above this point, the
threshold comes from the highest bid from other agents, as before. For
any bid $b$, we therefore define the \emph{expected threshold above
  $b$} to be $\thresholdragent{b}= \int_{\bidallocagent(b)}^{1}
\bidfi(z)\ dz$. With $b=\reserveagent$, the expected threshold above
$\reserveagent$ is precisely the portion of the expected threshold
generated by bids. See \Cref{fig:thresholdreserves} for an
illustration.

%Comparison of welfare and revenue approximation pictures.
	\begin{figure}[t]
			\small
			\centering
			%\hspace*{\fill}
			\subfloat[][In a first-price auction with reserve $\reserveagent$, the threshold above $\reserveagent$, $\thresholdragent{\reserveagent}$, only includes the thresholds when greater than $\reserveagent$, which corresponds to the case that the threshold is from a bid from another agent rather than the reserve.\label{fig:thresholdreserves} ]{
		\begin{tikzpicture}[xscale=3.5, yscale=3.5, domain=0:0.9, smooth]
		\draw[ pattern=north west lines, pattern color=lightgray] (0,0) -- (0,1) -- (1.12, 1) -- (.91,0.7) .. controls ( 0.68,0.4) and ( 0.6,0.32) .. (0.16,0.1) --(0,0.1);
		\node [fill=white] at ( 0.39,0.56) {\footnotesize $\thresholdragent{\reserveagent}$};
		\draw  (1.12,1) -- (.91,0.7) .. controls ( 0.68,0.4) and ( 0.6,0.32) .. (0.16,0.1) --(0.16,0);
	    \node at (1.2,0.9) {\footnotesize $\bidallocagent(\dev)$};
	    %axes
	    \draw[-] (0,0) -- (0,1) node[left] {$1$};
		\draw[-] (0,0) -- (1.4,0) node[below] {Bid ($\dev$)};
		%\node at (1, -0.09) {\footnotesize $\valagent$};
		%\draw[] (1, 0) -- (1, -0.025);
		\node at (0.16, -0.09) {\footnotesize $\reserveagent$};
		\draw[] (0.16, 0) -- (0.16, -0.025);				
		%\draw[pattern=north east lines, pattern color=lightgray] ( 0.5,0) rectangle ( 1,0.29);
		%\node [fill=white!50] at (0.75,0.15){\footnotesize $\bidutilagent(\bidagent)$};
		%\draw[ pattern=north west lines, pattern color=lightgray] (0,0) -- (0,0.7) -- (.91,0.7) .. controls ( 0.68,0.4) and ( 0.6,0.32) .. (0,0);
		%\node [fill=white] at ( 0.3,0.45) {$\thresholdagent(\allocalt,\strat)$};
		\node at (-0.16,0.3) {\footnotesize $\allocagent(\valagent)$};	
		\draw[] (-0.02, 0.29) -- (0, 0.29);	
		%\node at (-0.12, 0) {$\allocmin$};			
		%\node at (0.5,-0.09) {\footnotesize $\bidagent$};
		%\draw[] (0.5, 0) -- (0.5, -0.025);
		%\draw[line width=1.3pt, dashed] (0,0) rectangle (1,1);
		%\node at (-.1, 1.08) {Allocation};
		\node at (0.7,1.1) {Bid Allocation Rule};
	\end{tikzpicture}
	%REVENUE PICTURE - VIRTUAL VALUES
		}\IFECELSE{\hspace*{0.2cm}}{\hspace*{1cm}}
			\subfloat[][\Cref{lem:valcoveringreserves} shows that $\valagent \allocagent(\valagent)$ and $\thresholdragent{\reserveagent}$ cover an $(e-1)/e$ fraction of $i$'s  value $\valagent$, because $\valagent\allocagent(\valagent)-\utilagent(\valagent)=\bidagent\allocagent(\valagent)$ covers $\expthresholdagent-\thresholdragent{\reserveagent}$ \label{fig:vcreserves}]{
	\begin{tikzpicture}[xscale=3.5, yscale=3.5, domain=0:0.9, smooth]
		\draw[ pattern=north west lines, pattern color=lightgray] (0,0) -- (0,1) -- (1.12, 1) -- (.91,0.7) .. controls ( 0.68,0.4) and ( 0.6,0.32) .. (0.16,0.1) --(0,0.1);
		\node [fill=white] at ( 0.39,0.56) {\footnotesize $\thresholdragent{\reserveagent}$};
		\draw  (1.12,1) -- (.91,0.7) .. controls ( 0.68,0.4) and ( 0.6,0.32) .. (0.16,0.1) --(0.16,0);
	    \node at (1.2,0.9) {\footnotesize $\bidallocagent(\dev)$};
	    %axes
	    \draw[-] (0,0) -- (0,1) node[left] {$1$};
		\draw[-] (0,0) -- (1.4,0) node[below] {Bid ($\dev$)};
		\node at (1, -0.09) {\footnotesize $\valagent$};
		\draw[] (1, 0) -- (1, -0.025);
		\node at (0.16, -0.09) {\footnotesize $\reserveagent$};
		\draw[] (0.16, 0) -- (0.16, -0.025);				
		\draw[pattern=north east lines, pattern color=lightgray] ( 0,0) rectangle ( 1,0.29);
		\node [fill=white!50] at (0.5,0.15){\footnotesize $\valagent\allocagent(\valagent)$};
		%\draw[ pattern=north west lines, pattern color=lightgray] (0,0) -- (0,0.7) -- (.91,0.7) .. controls ( 0.68,0.4) and ( 0.6,0.32) .. (0,0);
		%\node [fill=white] at ( 0.3,0.45) {$\thresholdagent(\allocalt,\strat)$};
		\node at (-0.16,0.3) {\footnotesize $\allocagent(\valagent)$};	
		\draw[] (-0.02, 0.29) -- (0, 0.29);	
		%\node at (-0.12, 0) {$\allocmin$};			
		\node at (0.5,-0.09) {\footnotesize $\bidagent$};
		\draw[] (0.5, 0) -- (0.5, -0.025);
		\draw[line width=1.3pt, dashed] (0,0) rectangle (1,1);
		%\node at (-.1, 1.08) {Allocation};
		\node at (0.7,1.1) {Bid Allocation Rule};
	\end{tikzpicture}
	}
	\caption{}
\end{figure}
\noindent

%\subsection{Value Covering}

We now prove weaker notions of value and virtual value covering using
$\thresholdragent{\reserveagent}$ instead of
$\expthresholdagent$. Because $\thresholdragent{\reserveagent}\leq
\expthresholdagent$, value covering as stated in \Cref{lem:covering}
no longer holds. To solve this problem, we increase the lefthand side
by including an agent's expected payments. Meanwhile, revenue covering
still holds with $\thresholdragent{\reserveagent}$ instead of
$\expthresholdagent$. See \Cref{fig:vcreserves} for an
illustration. Formally:

\begin{lemma}[Value Covering with Reserves] \label{lem:valcoveringreserves}
In any BNE of $\FPAr$, for any bidder $i$ with value $\valagent\geq \reserveagent$, 
\begin{equation}
\valagent\allocagent(\valagent) + \thresholdragent{\reserveagent} \geq \tfrac{e-1}{e} \valagent. \label{eq:valcovering_res}
\end{equation}
\end{lemma} 
\begin{proof}
In BNE, if an agent's value is above the reserve, it is a best
response to bid at least the reserve. That is,
$\actionagent(\valagent)\geq \reserveagent$. If the reserve price is
ever binding for agent $\agent$, we have
$\thresholdragent{\reserveagent} + \paymentagent(\valagent) =
\thresholdragent{\reserveagent} +
\bidagent(\valagent)\allocagent(\valagent) \geq
\thresholdragent{\reserveagent} +\reserveagent
\bidallocagent(\reserveagent) = \expthresholdagent$. Otherwise, the
reserve never binds, and hence
$\thresholdragent{\reserveagent}=\expthresholdagent$. The result now
follows from applying \Cref{lem:covering} and the definition of
bidder utility.
\end{proof}

\begin{lemma}[Virtual Value Covering with Reserves]\label{lem:vvalres}
In any BNE of $\FPAr$, for any bidder $i$ with value $\valagent\geq
\reserveagent$ such that $\vvalagent(\valagent)\geq 0$,
\begin{equation}
\vvalagent(\valagent)\allocagent(\valagent) + \thresholdragent{\reserveagent} \geq \tfrac{e-1}{e} \vvalagent(\valagent). \label{eq:vvalcoveringeasyres}
\end{equation}
\end{lemma}  
\begin{proof}
Because $\valagent\geq\vvalagent(\valagent)$ for all values
$\valagent$, the result follows from \Cref{lem:valcoveringreserves}.
\end{proof}

Adapting revenue covering to accomodate reserves is simple. As
previously mentioned, the thresholds for bidder $\agent$ above
$\reserveagent$ correspond to bids from other agents. It follows that
this portion of $\agent$'s expected threshold corresponds to
revenue. We can formalize this with the following lemma:

\begin{lemma}[Revenue Covering with Reserves]\label{lem:rcres}
Fix an arbitrary bid distribution $\actionsdist$ of the first price
auction with reserves $\reserves$. For any feasible allocation
$\allocalt$,
\begin{equation}
\label{eq:FPARCr}
\REV(\FPAr,\actionsdist) \geq  \sum_{\agent}\thresholdragent{\reserveagent}\allocaltagent.
\end{equation}
\end{lemma}
\begin{proof}
For any agent $\agent$, recall that $\thresholdragent{\reserveagent}=
\int_{\bidallocagent(\reserveagent)}^{1} \bidfi(z)\ dz$. That is,
$\thresholdragent{\reserveagent}$ is the contribution to $\agent$'s
expected threshold bid from other agents' bids above
$\reserveagent$. Whenever $\agent$ faces a threshold from such a bid,
the mechanism makes at least as much revenue, as it charges the
highest bid. Hence for every agent $\agent$,
$\REV(\FPAr,\actionsdist)\geq\thresholdragent{\reserveagent}$. The
result follows.
\end{proof}

With regular value distributions, adding the monopoly reserves $\reserveagentm=\vvalagent^{-1}(0)$ to the auction excludes exactly the agents with negative virtual values. It follows that for such reserves, \eqref{eq:vvalcoveringeasyres} holds whenever $\valagent\geq\reserveagentm$. Moreover, the optimal mechanism for revenue allocates the item to the agent with the highest positive virtual value. To approximate the optimal revenue, it therefore suffices to approximate this agent's expected virtual value. By adapting revenue covering and virtual value covering to the first price auction with reserves, we are able to treat this quantity much as we treated welfare, yielding the following:

%we can now consider the first price auction with monopoly reserves, $\reserveagentm=\vvalagent^{-1}(0)$. 

%Without reserves, we used the notion of value covering to capture the intuition that bidder $i$'s expected threshold covered the gap between bidder $i$'s utility and their value. Because the expected threshold translated into revenue via revenue covering, we were able to prove a price of anarchy result. We wish to capture a similar value covering intuition after the addition of reserves, but to get a price of anarchy result, we again require that the expected threshold bids used correspond to agents' bids. In other words, we need a definition of value covering involving $\thresholdragent{\reserveagent}$.

\begin{theorem}\label{thm:revmonopeasy}
The revenue in any BNE of the first price auction with monopoly reserves and agents with regularly distributed values is at least a $2e/(e-1)$-approximation to revenue of the optimal auction.
\end{theorem}
\begin{proof}
We begin by summing inequality \eqref{eq:vvalcoveringeasyres} for each agent with $\valagent\geq \reserveagentm$ in an arbitrary value profile $\vals$:
\begin{equation}
\sum_{\agent:\valagent\geq \reserveagentm}\vvalagent(\valagent)\allocagent(\valagent)+\sum_{\agent:\valagent\geq \reserveagentm}\thresholdragent{\reserveagent}\geq\frac{e-1}{e}\sum_{\agent:\valagent\geq \reserveagentm}\vvalagent(\valagent).
\notag
\end{equation}
Let $\allocopt(\vals)$ be the allocation of the optimal mechanism on $\vals$. Since $\allocoptagent(\vals)\leq 1$ for each agent $\agent$, and since $\vvalagent(\valagent)\allocagent(\valagent)\geq 0$, we obtain:
\begin{equation}
\label{eq:summed}
\sum_{\agent:\valagent\geq \reserveagentm}\vvalagent(\valagent)\allocagent(\valagent)+\sum_{\agent:\valagent\geq \reserveagentm}\thresholdragent{\reserveagent}\allocoptagent(\vals)\geq\frac{e-1}{e}\sum_{\agent:\valagent\geq \reserveagentm}\vvalagent(\valagent)\allocoptagent(\vals).
\end{equation}

Both the first-price auction with monopoly reserves and the optimal
auction exclude agents with $\valagent<\reserveagentm$. Taking
expectations of \eqref{eq:summed}, the first term is the expected
revenue of the first-price auction with monopoly reserves
$\REV(\FPAm)$, and the sum on the right-hand side is the optimal
revenue, $\REV(\OPTm)$. Applying \Cref{lem:rcres} to the second term
on the left-hand side, therefore, yields $2\REV(\FPAm)\geq
\tfrac{e-1}{e}\REV(\OPTm)$, as desired.
\end{proof}

\subsection{Duplicate bidders}
\label{sec:duplicates}	

Another approach to mitigating the impact of negative virtual-valued
agents is to ensure each agent faces adequate
competition. \citet{BK96} show that this intuition guarantees
approximately optimal revenue in regular, symmetric, single-item
settings. In particular, their results can be interpreted as showing
that the second price auction with any reserve, even one which allows
agents with negative virtual values to be allocated, cannot have its
revenue too badly diminished by the contributions of low-valued
agents.

Formally, for any mechanism $\mech$, let $\REVpos(\mech) = \sum_i
\Ex[\valagent]{\max(0, \vvalagent(\valagent)) \allocagent(\valagent)}$
denote the expected positive virtual surplus of $\mech$. Given a
symmetric randomized reserve distribution $R$, let $\SPA_R$ denote the
second price auction with reserve $R$. A simple reinterpretation of
\citet{BK96} shows the following:

\begin{theorem}[\citealp{BK96}]
\label{thm:bk}
 For any symmetric randomized reserve with distribution $R$ and any single-item environment with $n$ i.i.d. regular bidders, the following inequality holds:
 \begin{equation}
 \REV(\SPA_R)\geq\tfrac{n-1}{n}\REVpos(\SPA_R).
 \end{equation}
\end{theorem}

%\citet{HR09} extend this intuition beyond single-item settings with \VCG\ in asymmetric settings: if every bidder must compete with a 
%duplicate, the revenue approximates the optimal in the unduplicated environment.

We show the same intuition holds for first-price and all-pay auctions
in asymmetric settings: if each bidder must compete with at least
$k-1$ other bidders with values drawn from her same distribution,
revenue is approximately optimal compared to the revenue-optimal
mechanism (including the duplicate bidders). We say such a setting
satisfies \emph{$k$-duplicates.} Formally:

\begin{definition}
A single-item environment satisfies \emph{$k$-duplicates} if the set
of agents can be partitioned into groups $B_1,\ldots,B_p$ for some
positive integer $p$ such that $|B_j|\geq k$ and the agents in $B_j$
are identically distributed, for each $j$ in $\{1,\ldots,p\}$.
\end{definition}

\newcommand{\FPAk}{\FPA_{k}}

We will generalize the analysis of \citet{BK96} to the first-price
auction ($\FPA$) and all-pay auction ($\APA$) with $k$-duplicates. The
first-price analysis will combine with the value and revenue covering
framework to produce a revenue approximation result. Moreover, in
\Cref{sec:allpay}, we extend the framework to include the
all-pay auction, which will yield a revenue result for that mechanism
as well.

\begin{lemma}
\label{lem:ourbk}
In any single-item setting with $k$-duplicates and regular value
distributions, the following inequalities hold:
\begin{align*}
\REV(\FPA)\geq \tfrac{k-1}{k}\REVpos(\FPA)\\
\REV(\APA)\geq \tfrac{k-1}{k}\REVpos(\APA)
\end{align*}
\end{lemma}

The proof reduces analyzing the allocation rule for each group of
duplicates to analyzing that of a second-price auction with a
randomized reserve generated by bidders outside the group. We may then
apply \Cref{thm:bk} and sum the virtual surplus from the different
groups. The full proof is included in \Cref{sec:app-rev}.

We can combine \Cref{lem:ourbk} with revenue covering and value
covering to derive a revenue bound for BNE of the first price auction,
which we state below:

\begin{theorem}\label{thm:fpadupes}
In any single-item environment with $k$-duplicates and regular value
distributions, the revenue in any BNE of the first price auction
($\FPA$) is at least a $\frac{k}{k-1}\frac{2e}{e-1}$ approximation to
the revenue of the optimal auction.
\end{theorem}

The proof is included in \Cref{sec:app-rev}. We discuss all-pay auctions in \Cref{sec:allpay} and derive similar revenue bounds.

\subsection{Revenue Lower Bounds}
For revenue, the approximation ratio of the first-price auction with
monopoly reserves can be at least as bad as $2$.  The same result was
shown by \citet{HR09} for the second-price auction with monopoly
reserves with the following two-agent example. One agent has a
deterministic value of $1$, the other agent has value drawn according
to the equal revenue distribution with support over $[1, H]$ for some
large $H$, with a light perturbation of the distribution so the
monopoly price is $1$. Assuming ties go to bidder $2$, an equilibrium
exists where both players bid $1$, giving revenue of $1$. The optimal
auction however can set a reserve of $H$ for the second bidder and
sell to the first bidder at price $1$ if the reserve is not met,
achieving a revenue of $2$ as $H$ grows.

%TODO - add note about the two auctions?
% In the single-item first-price and 
%
%Using a result of \citet{CH13}
%
%Additionally, we show that the first-price and all-pay auctions in which each bidder we use a result of \citet{CH13} to show that in a single-item setting, all bidders in a class that includes first-price and all-pay auctions (rank-and-bid based allocation rules, and bid-based payments) will behave symmetrically in BNE. As we showed the first-price auction is revenue covered in Section~\ref{sec:easy}, we can combine Lemma~\ref{lem:}
%Combining Lemma~\ref{lem:kbiddersnegvv} with Theorem~\ref{thm:revapx} and our earlier proof that the first-price auction is revenue covered then gives a revenue approximation result:
%
%\begin{corollary}
%\label{cor:corfpareve} The expected revenue in any BNE of a first-price single-item auction in a regular environment with at least 2 duplicates is a $3e/(e-1)$ approximation to the revenue of the optimal mechanism.
%\end{corollary}
%Proof included in the appendix of the full version of the paper.
%
%%\begin{corollary}
%%\label{cor:corallpayreve} The expected revenue in any BNE of an all-pay auction in a regular, matroid environment with at least 2 duplicates is a $4e/(e-1)$-APX to the revenue optimal mechanism.
%%\end{corollary}
%
%As the number of duplicates per group grows large, the approximation approaches $2e/(e-1)$.

\section{General Winner-Pays-Bid Mechanism}
\label{sec:beyondsingleitem}

We now extend the framework developed in Sections~\ref{sec:easy} and~\ref{sec:revenue} to more general feasibility environments, including matroids and single-minded combinatorial auctions. We will use our framework to prove welfare and revenue results for the analogues of the first-price auction, namely highest-bids-win winner-pays-bid mechanisms. As we will see, most of the proofs from the previous sections extend with little to no modification.

Two main ideas drove the single-item welfare (Theorem~\ref{thm:fpawelfare}) and revenue (Theorem~\ref{thm:revmonopeasy}) results. The first idea, value covering (resp.\ virtual value covering), captured the tradeoff between an agent's threshold bid and their utility (resp.\ virtual surplus). This idea depends only on a bidder's interim optimization problem, which is the same in every winner-pays-bid mechanism. We can extend the single-item proof to get:

\begin{lemma}[Pay-Your-Bid Value Covering]\label{lem:genvc}
In any BNE of any winner-pays-bid mechanism, for any bidder $\agent$ with value $\valagent$,
\begin{equation}
\label{eq:genvc}
\utilagent(\valagent)+\expthresholdagent \geq \tfrac{e-1}{e} \valagent.
\end{equation}
\end{lemma}

The second key ingredient is revenue covering, which captured the correspondence between threshold bids and mechanism revenue. Revenue covering is a mechanism-specific property, and must hold under every distribution of bids. In what follows, we will use a parameterized version of revenue covering in which the revenue needs only approximate the threshold bids.

\begin{definition}[$\revpar$-Revenue Covering]
\label{def:rc}
A mechanism $\mech$ is \emph{$\revpar$-revenue covered} if for any bid distribution $\actionsdist$ and feasible allocation $\allocalt$,
\begin{equation*}
\revpar \rev(\mech,\actionsdist)\geq\sum\nolimits_\agent \expthresholdagent\allocaltagent.
\end{equation*}
\end{definition}

Together, revenue covering and value covering provide a framework for deriving approximation results. We can combine revenue and value covering using the logic that drove the proof of Theorem~\ref{thm:fpawelfare}. Summing inequality (\ref{eq:genvc}) over all bidders, applying revenue covering, and taking expectations with respect to $\vals$ yields the following:

\begin{theorem}
\label{thm:pybwelf}
The welfare of any $\revpar$-revenue covered winner-pays-bid mechanism is a $\revpar\tfrac{e}{e-1}$-approximation to the welfare of any other mechanism.\footnote{Note that a tighter analysis of value-covering making use of the parameter $\revpar$ can give a bound of $\frac{\revpar}{1-e^{-\revpar}}$ (see \citet{ST13}). We do not include the analysis because it does not extend to results for revenue.}
\end{theorem}

The value and virtual value covering results for single-item auctions with reserves also hold without modification in general winner-pays-bid auctions. They require only value covering, winner-pays-bid semantics, and Myerson's virtual value characterization, all of which are agnostic to feasibility constraints. For example:

\begin{lemma}\label{lem:genvvcr}
In any BNE of any winner-pays-bid auction with reserves $\reserves$, for any bidder $\agent$ with value $\valagent\geq\reserveagent$ and $\vvalagent(\valagent)\geq 0$,
\begin{equation}
\vvalagent(\valagent)\allocagent(\valagent) + \thresholdragent{\reserveagent} \geq \tfrac{e-1}{e} \vvalagent(\valagent).
\notag
\end{equation}
\end{lemma}

Revenue analysis requires a way to translate thresholds into revenue. Revenue covering with reserves extends beyond single-item environments in the natural way. Specifically:

\begin{definition}[Revenue Covering with Reserves]\label{lem:rcresgen}
Given a vector of reserves $\reserves$, a mechanism $\mech$ is $\revpar$-revenue covered with reserves $\reserves$ if for any bid distribution $\actionsdist$ and feasible allocation $\allocalt$, 
\begin{equation}
\label{eq:FPARCgen}
\revpar \REV(M^\reserves, \actionsdist) \geq  \sum_{\agent}\thresholdragent{\reserveagent}\allocaltagent.
\end{equation}
\end{definition}

If a winner-pays-bid mechanism is $\revpar$-revenue covered without reserves, then the addition of any vector of reserves produces a mechanism which is $\revpar$-revenue covered with those reserves. Formally:
\begin{lemma}
\label{lem:rctorcr}
If a winner-pays-bid mechanism $\mech$ is $\revpar$-revenue covered, then adding any vector of individual reserves $\reserves$ produces a mechanism $\mech^\reserves$ that is $\revpar$-revenue covered with reserves $\reserves$.
\end{lemma}
\begin{proof}
Consider an arbitrary action profile $\actions$, and fix a vector of reserves $\reserves$. In what follows, let $\thresholdagent(\cdot)$ denote the threshold of agent $\agent$ under $\mech$ with no reserves. Now consider constructing an alternate bid profile $\actionsalt$ by setting the bid of any agent with $\actionagent<\reserveagent$ to $0$. Since $\actionsalt$ is a degenerate bid distribution, revenue covering implies:
\begin{equation}
\revpar \rev(\mech,\actionsalt)\geq\sum\nolimits_\agent \thresholdagent(\actionsaltothers)\allocaltagent.
\notag
\end{equation}
Note that the revenue of $\mech$ under $\actionsalt$ is the total bid of the bid-maximizing set of feasible agents whose bids exceed their reserves. In other words, it is exactly the revenue of $\mech^\reserves$ under $\actions$. Hence we have, for any bid distribution $\actionsdist$:
\begin{equation}
\Ex[\actions\sim\actionsdist]{\rev(\mech,\actionsalt)}=\rev(\mech^\reserves,\actionsdist).
\notag
\end{equation}
Next, consider the quantity $\Ex[\actionsaltothers\sim\actionsdist_{-\agentiter}]{\thresholdagent(\actionsaltothers)}$. This is the expected threshold of agent $\agent$ after the reserves of every other player have been applied, but with a reserve of $0$ for $\agent$. In other words, $\Ex[\actionsaltothers\sim\actionsdist_{-\agentiter}]{\thresholdagent(\actionsaltothers)}$ is the expected threshold bid of agent $\agent$ in $\mech^{(0,\reservesothers)}$ under action distribution $\actionsdist$. Note that considering only the contribution from threshold bids above $\reserveagent$, i.e.\ $\Ex[\actionsaltothers\sim\actionsdist_{-\agentiter}]{\thresholdagent(\actionsaltothers)\mathbbm 1_{\thresholdagent(\actionsaltothers)>\reserveagent}}$, produces a smaller quantity. Moreover, note that this latter quantity, the contribution to the $\agent$'s expected threshold of thresholds above $\reserveagent$ in $\mech^{(0,\reservesothers)}$, does not change if we consider $\mech^\reserves$ instead of $\mech^{(0,\reservesothers)}$. This latter quantity is $\thresholdragent{\reserveagent}$. Hence  for any alternate allocation $\allocalt$, we have:
\begin{align*}
\revpar\rev(\mech^\reserves,\actionsdist)&=\revpar\Ex[\actions\sim\actionsdist]{\rev(\mech,\actionsalt)}\\
&\geq\sum\nolimits_\agent\Ex[\actionsaltothers\sim\actionsdist_{-\agentiter}]{\thresholdagent(\actionsaltothers)}\allocaltagent\\
&\geq\sum\nolimits_\agent\Ex[\actionsaltothers\sim\actionsdist_{-\agentiter}]{\thresholdagent(\actionsaltothers)\mathbbm 1_{\thresholdagent(\actionsaltothers)>\reserveagent}}\allocaltagent\\
&=\sum\nolimits_\agent\thresholdragent{\reserveagent}\allocaltagent.\qedhere
\end{align*}
\end{proof}

Once $\revpar$-revenue covering with monopoly reserves has been established, it is a matter of extending the logic used to prove Theorem~\ref{thm:revmonopeasy} to produce a more general revenue result. We can sum, apply revenue covering, and take expectations to obtain:

\begin{theorem}
\label{thm:pybrev}
The revenue of any $\revpar$-revenue covered winner-pays-bid mechanism with regular bidders and monopoly reserves is a $(\revpar+1)\tfrac{e}{e-1}$-approximation to the revenue of the optimal mechanism.
\end{theorem}

In what follows, we show how to derive revenue covering results for environments beyond single-item auctions. In Section~\ref{sec:greedy}, we study general single-parameter feasibility environments, and show that greedy approximation algorithms are revenue covered. In Section~\ref{sec:posauctions} we discuss the revenue covering of the generalized-first-price position auction, and in Section~\ref{sec:discbids} we discuss winner-pays-bid auctions with a discretized bid space. In each section, welfare and revenue results will follow as corollaries to our analysis.

\subsection{Greedy Auctions}
\label{sec:greedy}

In a single-item auction, the bid of the unique winner gives both the revenue and the losers' threshold
bids. This immediately implied
$1$-revenue covering. For more general feasibility environments, the
relationship between total revenue and a particular agent's threshold
bid is less straightforward, and the revenue covering parameter
depends on the algorithm used to select winners. In this section we
show that mechanisms which allocate agents greedily have revenue
covering parameter equal to their approximation ratio.

We first formally define greedy algorithms. Intuitively, greedy
algorithms order agents, and then make a single pass over the
ordering, allocating each agent in turn if doing so maintains
feasibility. A greedy algorithm is therefore distinguished by the way
it selects an ordering.
\begin{definition}
	The \emph{greedy by priority} algorithm is given by a profile $\boldsymbol\psi=(\psi_1,\ldots,\psi_\numagents)$ of nondecreasing priority functions mapping bids for each agent $\agent$ to real numbers. It proceeds in the following way:
	\begin{enumerate}
		\item Sort agents in nonincreasing order of priority $\psi_\agent(\actionagent)$.
		\item Initialize the set of winners $S=\emptyset$.
		\item For each agent $\agent$ in sorted order: if $S\cup\{\agent\}$ is feasible, $S=S\cup\{\agent\}$.
		\item Return $S$.
	\end{enumerate}
\end{definition}
For example, the greedy-by-bid algorithm is given by priority functions $\psi_\agent(\actionagent)=\actionagent$ for all $\agent$. Our results will hold for general single-parameter feasibility environments, but we first highlight two settings where greedy algorithms are of particular interest: matroids and single-minded combinatorial environments. We define each below.

\begin{definition}
	A feasibility environment $\mathcal F$ is a \emph{matroid} if the following two properties hold:
	\begin{itemize}
		\item[i.] (Downward Closure) For any $S\in \mathcal F$ and $i\in S$, $S\setminus\{i\}\in \mathcal F$.
		\item[ii.] (Augmentation Property) For any $S_1,S_2\in \mathcal F$ with $|S_1|>
		 |S_2|$, there exists $i\in S_1\setminus S_2$ such that $S_2\cup\{i\}\in \mathcal F$.
	\end{itemize}
\end{definition}

Notable examples of matroids include \emph{$k$-unit environments}, where the seller has $k$ identical items to sell to buyers who demand at most one each, and \emph{transversal matroids}, which are matchable subsets of vertices on one side of a bipartite graph.

%\begin{lemma}
%	Let $\bidalloc(\actions)$ be the allocation selected by the greedy allocation rule with priority functions $\psi_\agent(\actionagent)=\actionagent$ under bid profile $\actions$, and $\allocalt$ be any other feasible allocation. Then
%	\begin{equation}
%	\sum\nolimits_\agent \actionagent\bidallocagent(\actions)\geq\sum\nolimits_\agent\actionagent\allocalt.
%	\notag
%	\end{equation}
%\end{lemma}

\begin{definition}
A \emph{single-minded combinatorial auction} feasibility environment
is defined by $\numitems$ indivisible items, $n$ agents that each desire
a bundle of items, and the constraint that no item can be allocated
more than once.  Agent $\agent$ desires the set of items
$\wantsagent$, she receives value $\valagent$ for receiving any
superset of $\wantsagent$ and value $0$ otherwise.  An allocation vector
$\alloc\in\{0,1\}^n$ is feasible if and only if for all agents $\agent \neq \agentalt$, simultaneous allocation $\allocagent = \allocagent[\agentalt] = 1$ implies disjoint demands $\wantsagent \cup \wantsagent[\agentalt] = \emptyset$.
\end{definition}

Maximizing the total value of winners is equivalent to the problem of
weighted set packing, which is NP-hard.  Moreover, and even for
special cases of the problem that are computationally tractable,
equilibria of highest-bid-wins winners-pay-bid mechanism can be a
factor $m$ from optimal (see \citealp{BL10}, and
Lemma~\ref{lem:combbad}, below).  As shown by \citet{BL10} and
rederived in our framework below, greedy algorithms are both
computationally tractable and give better welfare guarantees in
equilibrium.

\begin{definition}
An algorithm $\bidalloc$ is an $\alpha$-approximation for a feasibility environment $\mathcal F$ if for any bid profile $\actions$ and feasible allocation $\allocalt$, we have:
		\begin{equation}
	\sum\nolimits_\agent \actionagent\bidallocagent(\actions)\geq\frac{1}{\alpha}\sum\nolimits_\agent\actionagent\allocalt.
	\notag
	\end{equation}
\end{definition}

For matroids, the greedy-by-bid algorithm is optimal, i.e., a
$1$-approximation. For single-minded combinatorial environments, a
greedy algorithm that takes the size of each agent's desired bundle
into account achieves a $\sqrt{\numitems}$-approximation. We summarize
below.

\begin{lemma}
	For any matroid feasibility environment, greedy by priority with $\psi_\agent(\bidagent)=\bidagent$ for all agents $\agent$ is a $1$-approximation.
\end{lemma}

\begin{lemma}[\citealp{LOS02}]
	\label{lem:greedyapx}
	For any single-minded combinatorial auction environment, greedy by priority with $\psi_\agent(\bidagent)=\bidagent/\sqrt{|\wantsagent|}$ for all agents $\agent$ is a $\sqrt{\numitems}$-approximation.
\end{lemma}

\citet{BL10} show that any winner-pays-bid auction which allocates
using a greedy $\alpha$-approximation achieves at least a
$(\alpha+O(\alpha^2/e^\alpha))^{-1}$-fraction of the optimal welfare
in equilibrium. This result holds over a broad range of
multi-parameter settings. For single-parameter settings, we show that
this result is a simple consequence of revenue covering.

\begin{theorem}
	\label{thm:greedy}
	For any feasibility environment $\mathcal F$ and greedy $\alpha$-approximation algorithm $\mathcal A$ for $\mathcal F$, the winner-pays-bid mechanism which allocates according to $\mathcal A$ is $\alpha$-revenue covered.
\end{theorem}

For matroid and single-minded combinatorial environments, this immediately implies the following corollaries:

\begin{corollary}
	For any matroid feasibility environment, the winner-pays-bid mechanism with greedy by priority allocation with $\psi_\agent(\bidagent)=\bidagent$ for all agents $\agent$ is at least a $\tfrac{e}{e-1}$-approximation to the welfare of the optimal auction in any BNE. With monopoly reserves and regular value distributions, the mechanism is at least a $\tfrac{2e}{e-1}$-approximation to the revenue of the optimal auction.
\end{corollary}

\begin{corollary}
	\label{cor:comb}
	For any single-minded combinatorial environment, the winner-pays-bid mechanism with greedy by priority allocation with $\psi_\agent(\bidagent)=\bidagent/\sqrt{|\wantsagent|}$ for all agents $\agent$ is at least a $\sqrt{\numitems}\tfrac{e}{e-1}$-approximation to the welfare of the optimal auction in any BNE. With monopoly reserves and regular value distributions, the mechanism is at least a $(\sqrt{\numitems}+1)\tfrac{e}{e-1}$-approximation to the revenue of the optimal auction.
\end{corollary}

Before proving Theorem~\ref{thm:greedy}, we motivate the analysis with an example where a non-greedy algorithm with good performance in the absence of incentives fails to obtain high welfare in equilibrium, and therefore fails to be meaningfully revenue-covered.

\begin{definition}
	The \emph{winner-pays-bid highest-bid-wins mechanism} allocates the feasible set of bidders with the highest total bid, and charges winners their bid.
\end{definition}

\begin{lemma}
	\label{lem:combbad}
	There exists a single-minded combinatorial auction environment where the highest-bids-win mechanism is not $\revpar$-revenue covered for any $\revpar<\numitems$.
\end{lemma}

\begin{proof}
	Consider a setting with $\numitems$ items and $\numitems+2$ bidders. The first $\numitems$ bidders each want a single item: bidder $\mitem$ wants item $\mitem$, each with a value for allocation of 1, deterministically. The final two bidders, meanwhile, each want the grand bundle of all $\numitems$ items, with values of $1+\epsilon$, again deterministically. With appropriate tiebreaking, it is a BNE for each of the first $m$ agents to bid 0, while the final two bidders bid $1+\epsilon$. The optimal social welfare and revenue are both attained by selling to the first $\numitems$ bidders, for welfare and revenue of $\numitems$, yielding a factor of $\numitems$ loss in both welfare and revenue.
	
	This equilibrium also lower bounds the revenue-covering parameter for the highest-bids-win mechanism. First, note that the total revenue of the mechanism is $1+\epsilon$. Next, consider the feasible allocation $\allocalt$ which allocates the first $\numitems$ bidders. For these bidders, they must bid at least $1+\epsilon$ to get allocated, at which point they get allocated with probability 1. It follows that for such bidders, $\expthresholdagent=1+\epsilon$, and hence $\sum_\agent \expthresholdagent\allocaltagent=(1+\epsilon)\numitems$, $\numitems$ times the mechanism's revenue.
\end{proof}

In the example, the high bids of the $(\numitems+1)$th and
$(\numitems+2)$th bidders discouraged participation from the other
bidders - individually, each bidder would have needed to bid
$1+\epsilon$ to win. As a group, though, the losing bidders could have
won by increasing each of their bids by a tiny amount. Greedy
algorithms lack this pathology. On the other hand, for any greedy
allocation algorithm, we could increase the bids of every losing agent
to their threshold without changing the outcome. We formalize this
property as follows.

\begin{definition}
Allocation rule $\bidalloc$ is \emph{coalitionally non-bossy} if: for
any profiles of bids $\actions$ and $\actionsalt$ where the bids in
$\actionsalt$ are the same as $\actions$ for winners under $\actions$
and at most their critical prices for losers under $\actions$, i.e.,
if $\bidallocagent(\actions)=0$ then $\actionaltagent \leq
\thresholdagent(\actionothers)$; then the allocations of $\bidalloc$ under $\actions$
and $\actionsalt$ are the same
$\bidalloc(\actions)=\bidalloc(\actionsalt)$.
\end{definition}

\begin{lemma}\label{lem:bossy}
 	Any greedy by priority allocation rule is coalitionally non-bossy.
\end{lemma}
\begin{proof}
Imagine changing $\actions$ to $\actionsalt$ by increasing one loser's
bid at a time. Each time we increase a bid, say, of bidder $\agent$,
two things remain true: (1) $\agent$ still loses: as long as
$\actionaltagent\leq\thresholdagent(\actionothers)$, $\agent$ is
passed over as infeasible when she is reached by the greedy
algorithm; and (2) the threshold of every other losing agent $\agentalt$
remains unchanged: each losing agent's
threshold is only determined by the bids of the agents who win.
\end{proof}

\begin{proof}[Proof of Theorem~\ref{thm:greedy}]
	We argue that for any action profile $\actions$ and alternate allocation $\allocalt$,
	\begin{equation}
	\sum\nolimits_\agent \actionagent\bidallocagent(\actions) \geq \tfrac{1}{\alpha}\sum\nolimits_\agent \threshpt{\actionothers}\allocaltagent.
	\notag
	\end{equation}
	Taking the expectation of both sides yields the desired inequality.
	
	Let $\actionsalt$ be a vector of bids where losers under $\actions$ bid $\threshpt{\actionothers}$, while winners bid as before. The following inequalities hold, with justifications after.
	\begin{align*}
	\sum\nolimits_\agent \actionagent\bidallocagent(\actionothers) &= \sum\nolimits_\agent \actionaltagent\bidallocagent(\actionothers)&\\
	&= \sum\nolimits_\agent \actionaltagent\bidallocagent(\actionsalt)\\
	&\geq\tfrac{1}{\alpha}\sum\nolimits_\agent\actionaltagent\allocaltagent\\
	&\geq\tfrac{1}{\alpha}\sum\nolimits_\agent \threshpt{\actionothers}\allocaltagent.
	\end{align*}
	The first line holds because $\actionsalt$ differs from $\actions$ only on the bids of losing agents. The second follows from Lemma~\ref{lem:bossy}, and the third from Lemma~\ref{lem:greedyapx}. The last line follows from the fact that $\actionsalt$ doesn't change the bids of winners under $\actions$, and for those agents, $\actionagent\geq\threshpt{\actionothers}$.
\end{proof}

The conclusions of Theorem~\ref{thm:greedy} and its analysis are twofold. First, comparing the greedy-by-priority algorithm from Lemma~\ref{lem:greedyapx} to the highest-bids-win algorithm of Lemma~\ref{lem:combbad} reveals that welfare loss in equilibrium can stem from two sources. Ignoring computational constraints, the highest-bids-win algorithm is optimal in the absence of incentives. Hence, the factor of $\Omega(m)$ welfare loss in the example of Lemma~\ref{lem:combbad} is entirely due to incentives. The greedy algorithm, on the other hand, produces suboptimal allocations without incentives. However, Corollary~\ref{cor:comb} states that the loss from the introduction of incentives is limited to at most a factor of $e/(e-1)$. Hence, the welfare loss is primarily due the algorithm's performance without incentives. Second, this comparison suggests the value of revenue covering as objective in algorithm design. Indeed, \citet{DK15} subsequently study the design of algorithms with low revenue covering parameters. 
For single-minded combinatorial auctions, they show that no algorithm is $o(\sqrt{\numitems})$-revenue covered. Hence, the greedy algorithm is optimal with respect to this objective.

\subsection{Position Auctions}\label{sec:posauctions}
%\TODO{Put in Appendix}

In this section, we derive welfare and revenue results for winner-pays-bid position auctions, a.k.a., the generalized first-price auction (GFP). While these results are driven by the same fundamental notions as before, value covering and revenue covering, the precise formulations of these two ideas must be modified to fit environments such as position auctions where allocation is inherently probabilistic. We give the appropriate definitions of value covering and revenue covering that yield the results, and relegate the proof details to Appendix~\ref{sec:appGFP}.

Formally, a position auction is an auction in which agents can be allocated one of $m$ positions; position $j$ is valued by an agent at $\alpha_j \valagent$. In advertising auctions, these are positions on a web-page to fill where lower positions receive fewer clicks. The positions are ordered such that $\{\alpha_j\}$ is decreasing in $j$ (hence position 1 is best). In GFP, agents submit bids $\bid_i$, and positions are allocated in order of bid. Each agent pays their bid scaled by the quality of the position:  $\alpha_j \bid_i$. Equivalently, they pay their bid when they are served, which occurs with probability $\alpha_j$ for position $j$.

In a position auction, an agent can win one of several different positions, and consequently the agent's threshold bid for guaranteed allocation is not useful for analysis. We use instead a finer-grained measure of an agent's threshold bid that is parameterized by allocation probability. In particular, we consider $T_i(\allocaltagent) = \int_0^{\allocaltagent} \ppu(z)\ dz$, which we refer to $T_i(\allocaltagent)$ as the \emph{partial threshold}.

\begin{definition}
For any (implicit) strategy profile $\strat$ and any allocation probability $\allocaltagent\in[0,1]$, the \emph{partial threshold} for $\allocaltagent$, denoted $T_i(\allocaltagent)$, is given by $T_i(\allocaltagent) = \int_0^{\allocaltagent} \ppu(z)\ dz$.
\end{definition}

With partial thresholds defined, we can now state revenue covering for partial thresholds:

\newcommand{\auction}{A}
\newcommand{\threshgfp}[1]{T(#1)}
\newcommand{\threshgfpi}[1]{T_i(#1)}
\newcommand{\threshgfpactions}[2]{T^{#2}(#1)}
\begin{definition}[Revenue Covering with Partial Thresholds]
\label{def:gfprevc}
A mechanism is $\mu$-revenue covered with partial thresholds if for any (implicit) strategy profile $\strat$, and feasible allocation $\allocalt$, 
\begin{equation}
\revpar \REV(\mech) \geq \sum\nolimits_i \threshgfpi{\allocaltagent}.
\end{equation}
\end{definition}
 %By Lemma~\ref{lem:fpaaddingreserves}, adding reserves preserves revenue covering above the reserves and so Theorem~\ref{thm:poa} and Lemma~\ref{lem:monopreserves} respectively imply welfare and revenue approximations of $2e/(e-1)$ with reserves. 

\begin{theorem}\label{thm:gfp-covered}
	The generalized first price auction is $1$-revenue covered with partial thresholds.
\end{theorem}

 The proof is included in Appendix~\ref{sec:appGFP}. Note that revenue covering with partial thresholds is weaker than the general revenue covering (Definition~\ref{def:rc}) condition. This follows from the fact that $\threshgfpi{\allocaltagent}$ is convex in $\allocaltagent$, and hence $\threshgfpi{\allocaltagent}\geq \expthresholdagent\allocaltagent$. With a weaker notion of revenue covering, we need a stronger value covering condition with partial thresholds.

\begin{lemma}[Value Covering for Partial Thresholds]\label{lem:vcpositionauctions}
	In any BNE of a winner-pays-bid auction, for any bidder $i$ with value $\valagent$, 
	%\allocoptagent(\valagent)\REV(\FPA)
	\begin{equation}
	\utilagent(\valagent)+\threshexpgfp{}{0}{\allocaltagent} \geq \tfrac{e-1}{e} \allocaltagent\valagent
	\label{eq:vcpositionauctions}.
	\end{equation}
\end{lemma}  
The proof is omitted; it proceeds almost identically to Lemma~\ref{lem:genvc} but with $\threshexpgfp{}{0}{\allocaltagent}$ in place of $T_i$ and $\allocaltagent\valagent$ in place of $\valagent$. The worst case $\utilagent$ becomes $\allocaltagent\valagent/e$, rather than $\valagent/e$. Combining revenue and value covering then gives a welfare approximation:

\begin{theorem}
The welfare of any BNE of $\GFP$ is an $\frac{e}{e-1}$-approximation to the welfare of the optimal auction.
\end{theorem}

All analysis prior to this section could be performed with partial thresholds rather than full thresholds $\expthresholdagent\allocaltagent$. We chose not to do so as the partial threshold definition is only necessary in probabilistic feasibility settings. In particular, the value and revenue covering framework with partial thresholds naturally extends to virtual values and reserves. This extension yields the following revenue result:

\begin{theorem}
For regular environments, the revenue in any BNE of $\GFP$ with monopoly reserves is a $\frac{2e}{e-1}$-approximation to the revenue of the optimal auction.
\end{theorem}

 %We obtain the following welfare and revenue results:
 
 %\begin{corollary}
% The welfare of the generalized first price auction is a $\tfrac{e}{e-1}$-approximation to that of any other mechanism.
% \end{corollary}
 
 %\begin{corollary}
 %The with regular value distributions, the revenue of the generalized first price auction with monopoly reserves is a $\tfrac{2e}{e-1}$-approximation to that of any other mechanism.
% \end{corollary}

\subsection{Discretized Bids}\label{sec:discbids}
In practice, the bid space in an auction is often discretized for convenience or feasibility. We note here that discretizing the bid space of a winner-pays-bid auction only results in a small additive loss to the bounds proved via the revenue covering framework. In particular, consider a winner-pays-bid mechanism with unrestricted bid space. Restricting the bid space to  integral multiples of $\delta$ for some $\delta >0$ preserves $\revpar$-revenue covering up to an additive term. Formally:

\begin{lemma}
\label{LEM:DISCRETIZED}
For any winner-pays-bid mechanism $\mech$, let $\mech^\delta$ be $\mech$ with bid space restricted to integral multiples of $\delta$. Let $\expthresholdagent^\delta$ be agent $\agent$'s expected threshold in $\mech^\delta$ for some bid distribution $\actionsdist$ over multiples of $\delta$. If $\mech$ is $\revpar$-revenue covered, then the following inequality holds for any feasible allocation $\allocalt$:
\begin{equation}
\rev(\mech^\delta,\actionsdist)\geq \sum\nolimits_\agent \expthresholdagent^\delta\allocaltagent-\numagents\delta.
\notag
\end{equation}
\end{lemma}
\begin{proof}
Assume agent $\agent$ bidding $\actionagent$ in $\mech$ against $\actionsdistothers$ gets allocation probability $\bidallocagent(\actionagent)$. To get allocation at least $\bidallocagent(\actionagent)$ in $\mech^\delta$ when facing $\actionsdistothers$, $\agent$ will need to bid at most the next-highest increment of $\mech$. Hence, if we define $\bidfi^\delta(\allocation)$ the smallest bid that achieves allocation of at least $\allocation$ in $\mech^\delta$, we have $\bidfi^\delta(\allocation)\leq \bidfi(\allocation)+\delta$ for all $\allocation\in[0,1]$. Integrating over all $\allocation\in[0,1]$ for each agents and applying $\mu$-revenue covering yields the result.
\end{proof}

The additive $\numagents\delta$ loss in revenue covering translates directly to an additive $\numagents\delta$ loss in revenue. If $\delta$ is small compared with the revenue of the original mechanism, this makes effectively no difference: if $\delta$ is very large relative to revenue, it could make a large difference.\footnote{In a setting where at most $m$ agents may be allocated simultaneously, the loss in Lemma~\ref{LEM:DISCRETIZED} becomes $m\delta$.}

%TODO - maybe do math in terms of the ratio of the discretization factor to values?
% does this work?

\begin{figure}[t]
	\centering
	\begin{tikzpicture}[xscale=3.5, yscale=3.5, domain=0:0.9, smooth]
	\draw  (1.12,1) -- (.91,0.7) .. controls ( 0.68,0.4) and ( 0.6,0.32) .. (0.16,0.1) --(0,0);
	\draw[pattern=north west lines, pattern color=lightgray] (0,0) -- (0.1, 0) -- (0.1, 0.06) -- (0.2, 0.06) -- (0.2, 0.12) -- (0.3, 0.12) -- (0.3, 0.17) -- (0.4, 0.17) -- (0.4, 0.22) -- (0.5, 0.22) -- (0.5, 0.29) -- (0.6, 0.29) -- (0.6, 0.36) -- (0.7, 0.36) -- (0.7, 0.45) -- (0.8, 0.45) -- (0.8, 0.56) -- (0.9, 0.56) -- (0.9, 0.69) -- (1, 0.69) -- (1, 0.83) -- (1.1, 0.83) -- (1.1, 0.97) -- (1.2, 0.97) -- (1.2, 1) -- (0,1) -- (0,0);
	\node at (1.2,0.9) {\footnotesize $\bidallocagent(\dev)$};
	\draw[-] (0,0) -- (0,1) node[left] {$1$};
	\draw[-] (0,0) -- (1.4,0) node[below] {Bid ($\dev$)};
	\node at (1, -0.09) {\footnotesize $\valagent$};
	\draw[pattern=north east lines, pattern color=lightgray] ( 0.5,0) rectangle ( 1,0.29);
	\node [fill=white!50] at (0.75,0.15){\footnotesize $\bidutilagent(\bidagent)$};
	\draw (0,0) -- (0,1) -- (1.12, 1) -- (.91,0.7) .. controls ( 0.68,0.4) and ( 0.6,0.32) .. (0.16,0.1) --(0,0);
	\node [fill=white] at ( 0.35,0.54) {\footnotesize $ \expthresholdagent$};
	\node at (-0.15,0.3) {\footnotesize $\bidallocagent(\bidagent)$};
	\draw[] (-0.02, 0.29) -- (0, 0.29);		
	\node at (0.5,-0.09) {\footnotesize $\bidagent$};
	\draw[] (1, 0) -- (1, -0.025);
	\draw[] (0.5, 0) -- (0.5, -0.025);
	\node at (0.7,1.1) {Bid Allocation Rule};
	\draw[] (0, -0.02) -- (0, -0.08);
	\draw[] (0.1, -0.02) -- (0.1, -0.08);
	\draw[] (0, -0.05) -- (0.1, -0.05);
	\node at (0.05, -0.15) {$\delta$};
	\draw (0,0) -- (0.1, 0) -- (0.1, 0.06) -- (0.2, 0.06) -- (0.2, 0.12) -- (0.3, 0.12) -- (0.3, 0.17) -- (0.4, 0.17) -- (0.4, 0.22) -- (0.5, 0.22) -- (0.5, 0.29) -- (0.6, 0.29) -- (0.6, 0.36) -- (0.7, 0.36) -- (0.7, 0.45) -- (0.8, 0.45) -- (0.8, 0.56) -- (0.9, 0.56) -- (0.9, 0.69) -- (1, 0.69) -- (1, 0.83) -- (1.1, 0.83) -- (1.1, 0.97) -- (1.2, 0.97) -- (1.2, 1);
	\end{tikzpicture}
	\caption{
		When bids are only allowed in multiples of $\delta$, the expected threshold expands, but only by at most an additive amount of $\delta$. \label{fig:discbids}}
\end{figure}
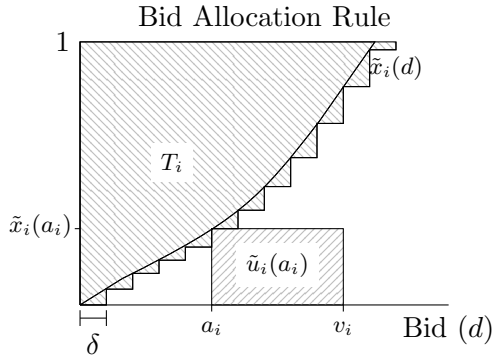

\section{Beyond Winner-pays-bid Auctions}
\label{sec:revenuecovering}
\label{sec:beyondpyb}

In this section, we extend our framework for proving worst-case
approximation bounds to auctions other than those with winner-pays-bid
semantics. We do this by showing that a bidder's decision problem in
any auction is equivalent to choosing how to bid in a winner-pays-bid
auction. Consequently, a form of value covering holds for \emph{any}
auction for single-dimensional agents with independent value
distributions in BNE. We then apply this winner-pays-bid reduction to
analysis of welfare and revenue in all-pay auctions
(\Cref{sec:allpay}), as well as to the simultaneous composition of
revenue-covered auctions (\Cref{sec:simul}).

\subsection{Extending the Framework}
\label{sec:frameworkpoa}

To instantiate the revenue covering framework in mechanisms with
general payment semantics, we first write the utility
$\bidutilagent(\cdot)$ of an agent $\agent$ as:

\begin{equation}
\bidutilagent(\actionagent) = \valagent\bidallocagent(\actionagent) - \bidpaymentagent(\actionagent) = \left(\valagent - \frac{\bidpaymentagent(\actionagent)}{\bidallocagent(\actionagent)}\right)\bidallocagent(\actionagent).
\notag
\end{equation}
Note the resemblance to the utility of an agent in a winner-pays-bid
mechanism: the term
$\frac{\bidpaymentagent(\actionagent)}{\bidallocagent(\actionagent)}$
plays the same role as the bid in a pay-your-bid auction: it is the
price per unit of allocation. For this reason, we will refer to it as
the equivalent bid for action $\actionagent$.

\begin{definition}
\label{def:equivbid}
The \emph{equivalent bid} for an action $\actionagent$, denoted
$\equivbidact{\actionagent}$, is given by
$\equivbidact{\actionagent}=\frac{\bidpaymentagent(\actionagent)}{\bidallocagent(\actionagent)}$.
\end{definition}
Note that unlike an agent's winner-pays-bid bid, the equivalent bid is
generally a function of other agents' bid distributions. We will now
define expected thresholds, revenue covering, and value covering with
equivalent bids playing the role of winner-pays-bid bids.

%By using the equivalent bid in place of the first-price bid

 %utility and expected (equivalent bid) threshold always approximate the value of a bidder, and if the revenue covers the expected (equivalent bid) thresholds then we can attain welfare and revenue approximation results. We illustrate this approach in Section~\ref{sec:allpay} for auctions with all-pay semantics and in Section~\ref{sec:simul} to show closure under simultaneous composition with unit-demand bidders.

\paragraph{Equivalent Threshold Bids}
\label{sec:equivthresholdbid}
%For first-price auctions, this is exactly the bid. For mechanisms with different payment semantics, $\equivbidact{\actionagent}$ can still be thought of as an equivalent first-price bid for action $\actionagent$.

The key quantity in our winner-pays-bid proof framework is the
threshold bid - the lowest value an agent must bid to get
allocated. Threshold bids are an ex-post notion, while equivalent bids
as given in \Cref{def:equivbid} are only well-defined in the
interim. Recall that in \Cref{sec:easy} we computed the expected
threshold bid by integrating the inverse of its cumulative
distribution function. We employ a similar approach for our more
general framework.

Recall that the bid threshold $\bidfi(\allocdev)$ for winner-pays-bid
mechanisms was the lowest bid which achieved allocation at least
$\allocdev$. We extend this to general mechanisms by using the lowest
equivalent bid required to achieve allocation at least $\allocdev$:

%of the CDF corresponded to the lowest bid which achieved allocation at least $\allocdev$. 

%The shape of the CDF, given by this function, was exactly the set of possible tradeoffs the agent could make between allocation probability and price per unit.

%In the first-price auction, the mapping between actions (i.e. bids) and price per unit of allocation was a bijection. This is not necessarily the case for general mechanisms. When there are multiple actions inducing the same price per unit of allocation, an agent who is best-responding will only consider those actions on the Pareto frontier of the tradeoff between price per unit of allocation and allocation probability. In terms of the allocation probability, we will formally define the frontier in the following way:

\begin{definition}
Fixing an action distribution $\actionsdist$ and allocation probability $\allocdev$, let $\effaction(\allocdev)=\argmin_{\actionagent:\bidallocagent(\actionagent)\geq \allocdev} \equivbidact{\actionagent}$ be the action inducing allocation probability at least $\allocdev$ with the lowest equivalent bid. The \emph{equivalent threshold bid} for $\allocdev$ is then given by $\ppu(\allocdev)=\equivbidact{\effaction(\allocdev)}$.
\end{definition}

%The equivalent threshold bid $\ppu(\allocdev)$ depends on the strategy profile $\strat$, as it is taken in expectation over the actions of the other agents; we suppress the strategy profile as an argument for notational convenience. Note that for the first-price auction, the equivalent threshold bid $\ppu(\allocdev)$ corresponds exactly to the inverse of the bid allocation rule.

In the first-price auction, we used value covering and revenue
covering to relate expected threshold bids to the welfare and revenue
of the first-price auction and the optimal auction. By taking the
expectation of the threshold bid, we were able to aggregate all of the
possible allocation versus price-per-unit tradeoffs an agent might
consider into a single parameter. In the more general setting, this
expectation is not defined, so we instead recall the alternative
formulation: we integrated the inverse of the bid allocation rule,
$\bidfi(\allocdev)$. For general mechanisms, we will aggregate
equivalent threshold bids in a similar manner. Formally:

\begin{definition}
Fixing an action distribution $\actionsdist$, let the \emph{cumulative
  equivalent threshold bid}, or simply \emph{cumulative threshold} for
agent $\agent$ to be
$\threshexpected{\bidallocagent(\actionagent)}{\alloclevel}=\int_{0}^{1}\ppu(\allocdev)\,d\allocdev.$
\end{definition}

We use the term ``cumulative threshold'' rather than ``expected
threshold'' in the framework for general mechanisms because unlike in
winner-pays-bid mechanisms, $\expthresholdagent$ does not necessarily
correspond to an expectation.

\paragraph{Covering Conditions and the Price of Anarchy.}

We now show how equivalent bids reduce the optimization problem of a
bidder in a general auction to that of a bidder in a first-price
auction. In particular, value covering, which quantifies the tradeoff
between utility and expected threshold, still holds with cumulative
thresholds:

%We now use equivalent bids and thresholds in place of first-price bids and thresholds to develop general analogues of the value and revenue covering conditions of Section~\ref{sec:easy}. 

\begin{lemma}[Value Covering]
\label{lem:vc}
Consider a mechanism $\mech$ in BNE and interim utility function
$\utilagent(\valagent)$ for agent $\agent$ with value $\valagent$ and
cumulative equivalent threshold bid $\expthresholdagent$. Then
\begin{equation}
\label{eq:vcl}
\utilagent(\valagent)+\expthresholdagent \geq \tfrac{e-1}{e}\valagent.
\end{equation}
\end{lemma}

The proof (included in \Cref{sec:appframework}) is identical to that
of winner-pays-bid value covering (\Cref{lem:covering}) because agent
$\agent$ is making the same tradeoff between allocation probability
and price per unit of allocation in each.

We define $\revpar$-revenue covering in the same way using the
generalized definition of cumulative equivalent threshold
$\expthresholdagent$:

\begin{definition}[$\revpar$-Revenue Covering]
A mechanism $\mech$ is \emph{$\revpar$-revenue covered} if for any
(implicit) distribution of actions $\actionsdist$ and feasible
allocation $\allocalt$,
\begin{equation*}
\revpar \rev(\mech)\geq\sum\nolimits_\agent \expthresholdagent\allocaltagent.
\end{equation*}
\end{definition}

As before, value covering and revenue covering together imply a welfare result:

\begin{theorem}[Extension of Theorem~\ref{thm:pybwelf}]
\label{thm:poa}
If a mechanism is $\revpar$-revenue covered, then in any BNE it is a $\mu \frac{e}{e-1}-$approximation to the welfare of the optimal mechanism.
\end{theorem}

Finally, virtual value covering, which held as a corollary of value covering, did not require any properties of $\expthresholdagent$, and consequently holds for general mechanisms as well:

\begin{lemma}[Extension of Lemma~\ref{lem:vvalcoveringeasy}]
\label{lem:vvcgeneral}
In any BNE of the first price auction, for any bidder $i$ with value $\valagent$ such that $\vvalagent(\valagent)\geq 0$, 
\begin{equation}
\vvalagent(\valagent)\allocagent(\valagent) + \threshexpected{\allocagent(\valagent)}{\alloclevel} \geq \tfrac{e-1}{e} \vvalagent(\valagent).
\end{equation}
\end{lemma}

%Note that the statement of Lemma~\ref{lem:vc} is identical to that of Lemma~\ref{lem:covering} - the only difference is the definition of the quantity $\expthresholdagent$. Large portions of the analysis subseqent to Lemma~\ref{lem:covering} made no use of the definition of $\expthresholdagent$. Consequently, these results extend without modification. We list these results below. The proofs are identical to their analogs in earlier sections.

%First, recall that when thresholds corresponded to revenue, i.e.\ revenue covering held, we were able to leverage the tradeoff between utility and thresholds into a welfare guarantee. The same is true in our general framework. We first state revenue covering:

%This statement is identical to Definition~\ref{def:rc}. As before, it implies a welfare result:

%The other key result which extends without modification is the adaptation of value covering to virtual surplus. In particular:

Note that \Cref{lem:vvcgeneral} does not directly give a revenue
result, as the techniques for controlling negative virtual surplus
(e.g.\ reserves) may not apply for all single-dimensional
mechanisms. In \Cref{sec:allpay}, however, we will show that with
sufficient competition, i.e.\ $k$-duplicates, the single-item all-pay
auction approximates the optimal revenue.

%Recall that for the first-price auction, revenue and value covering combined to give approximation results for welfare. The same holds for value covering and $\revpar$-revenue covering from Definition~\ref{def:rc} of general mechanisms.

%\begin{theorem}
%\label{thm:poa}
%If a mechanism is $\revpar$-revenue covered, then in any BNE it is a $\mu \frac{e}{e-1}-$approximation to the welfare of the optimal mechanism.
%\end{theorem}

%\begin{proof}
%The proof proceeds analogously to the proof of Theorem~\ref{thm:fpawelfare}.

%For any agent $i$ with value $\valagent$, value covering (Lemma~\ref{lem:vc}) and $\revpar$-revenue covering (Definition~\ref{def:rc}) combine to give
%\begin{equation*}
%\utilagent(\valagent)+\revpar\,\REV(\mech) 
	%\geq\tfrac{e-1}{e}\valagent.
%\end{equation*}

%\noindent
%Let $\allocopt(\vals)$ be the welfare-optimal allocation. For any agent, $\allocoptagent\leq 1$, hence

%\begin{equation}
%\utilagent(\valagent)+\allocoptagent(\valagent)\revpar\REV(\mech) \geq\allocoptagent(\valagent)\utilagent(\valagent)+\allocoptagent(\valagent)\revpar\REV(\mech) \geq \allocoptagent(\valagent)\frac{e-1}{e}\valagent.
%\end{equation}

%Summing over all agents and taking expectation over values gives $\UTIL(\mech) + \revpar \REV(\mech) \geq \frac{e-1}{e} \WEL(\OPTm)$. As $\UTIL(\mech) + \REV(\mech) = \WEL(\mech)$ and $\revpar \geq 1$, we then have our desired result,
%\begin{equation*}
%\revpar\frac{e}{e-1} \WEL(\mech) \geq \WEL(\OPTm).\qedhere
%\end{equation*}
%\end{proof}

\subsection{All-Pay Mechanisms}
\label{sec:allpay}

In an \emph{all-pay} mechanism, each agent pays their bid, regardless
of whether they win or lose.  We now show, by giving a reduction to
the analysis of winner-pays-bid mechanisms, that all-pay auctions have
approximately optimal welfare and, with sufficient competition,
revenue.  Since value covering holds for any mechanism in BNE, we must
simply show that all-pay mechanisms are revenue covered, i.e., the
expected revenue of the mechanism must approximate the cumulative
equivalent thresholds. Note that the equivalent bid corresponding to
an all-pay bid can be found simply by dividing by the allocation
probability: 
\begin{align}
\label{eq:allpay-equivalent-bid}
\equivbidact{\bidagent} &= \bidagent / \bidallocagent(\bidagent).
\end{align}

With a factor two loss we will reduce revenue covering in all-pay
mechanisms to revenue covering in winner-pays-bid mechanisms with the
same bid-allocation rule.  We conclude below that for single-item and
matroid environments, highest-bids-win all-pay mechanisms satisfy
revenue covering with $\mu=2$ and, thus, are $2e/(e-1)$ approximately
efficient.

To make the reduction clear we denote the cumulative effective
thresholds of all-pay and winner-pays-bid mechanisms with allocation
rule $\bidalloc$, respectively, as:
\begin{align*}
\threshexpectedallpay{0}{1} &= \int_0^1
\effbidallpay(\allocdev)\,d\allocdev,\\ \threshexpectedwinnerpaysbid{0}{1}
&= \int_0^1 \effbidwinnerpaysbid(\allocdev)\,d\allocdev,
\end{align*}
where $\effbidallpay(\allocdev)$ and $\effbidwinnerpaysbid(\allocdev)$
are the inverses of the respective effective bid allocation rules.
Note that the bid allocation rule and the effective bid allocation
rule are identical for winner-pays-bid mechanisms but distinct for
all-pay mechanisms.

\begin{lemma} 
\label{l:allpay-vs-winnerpaysbid}
For a common bid distribution, the cumulative equivalent thresholds of
all-pay and winner-pays-bid mechanisms satisfy
$\threshexpectedallpay{0}{1} \leq 2
\threshexpectedwinnerpaysbid{0}{1}$.
\end{lemma}

\begin{proof} 
Fix a common distribution of bids.  The effective thresholds of
all-pay and winner-pays-bid format mechanisms are related by
equation~\eqref{eq:allpay-equivalent-bid} as
\begin{align}
\label{eq:bvtbound}
\effbidallpay(\allocdev) &= \allocdev/\effbidwinnerpaysbid(\allocdev).
\end{align}
This relation yields the following sequence of inequalities:
\begin{equation}
\label{eq:bidtothresh}
\threshexpectedwinnerpaysbid{0}{1} = \int_0^1
\effbidwinnerpaysbid(\allocdev)\, d\allocdev =
\int_0^1\allocdev\,\effbidallpay(\allocdev)\,d\allocdev\geq
\frac{1}{2} \int_0^1 \effbidallpay(\allocdev)\,
d\allocdev=\frac{1}{2}\threshexpectedallpay{0}{1}
\end{equation}
where the second equality follows from equation~\eqref{eq:bvtbound}
and the inequality from Chebyshev's sum inequality and the fact that
$\effbidallpay$ is a non-decreasing function.
\end{proof}

\begin{theorem}
\label{lem:aprc}
For bid allocation rule $\bidalloc$, if the winner-pays-bid mechanism
is $\mu$-revenue covered, then the all-pay mechanism is $2\mu$-revenue
covered.
\end{theorem}

\begin{proof}
For a fixed common distribution of bids and common bid allocation rule
$\bidalloc$, the revenue of the all-pay mechanism exceeds the revenue
of the winner-pays-bid mechanism.  The former is the sum of all the
bids and the latter is the only the sum of the winning bids.  On the
other hand, \Cref{l:allpay-vs-winnerpaysbid} shows that the cumulative
effective threshold of the all-pay mechanism is at most twice that of
the winner-pays-bid mechanism for the common distribution of bids.
Thus, if $\mu$-revenue covering holds for the winner-pays-bid
mechanism, $2\mu$-revenue covering holds for the all-pay mechanism.
\end{proof}

Combining revenue covering with \Cref{thm:poa} for highest-bids-win
all-pay mechanisms in single-item and matroid environments gives a
welfare bound of $2e/(e-1)$. For revenue, using the techniques of
\Cref{sec:duplicates} and ensuring that at least two bidders have
values drawn from each distribution (i.e.\ $2$-duplicates) gives a
$4e/(e-1)$-approximation to the revenue of the optimal
auction.\footnote{We do not discuss a revenue approximation bound for
  all-pay mechanisms with reserves as exogenously imposed reserves in
  bid space map to endogenous reserves in value space.  This
  phenomenon makes imposing monopoly reserves in value space difficult
  with all-pay mechanisms.}

The bounds can be improved to $2$ and $6$ for welfare and revenue,
respectively, by adapting the value-covering condition to the all-pay
format, shown in \Cref{sec:allpayapp}.

\subsection{The Second-price Auction}
\label{sec:spa}

Not all mechanisms are revenue covered. In the second-price auction,
agents submit sealed bids, the highest bidder wins and is charged the
second-highest bid. This auction lacks a direct connection between
bidders' threshold bids and the revenue of the auction, which is
required for revenue covering. To illustrate, consider a two-agent
setting, and assume agent~1 bids 10 and agent~2 bids $0$,
deterministically. The revenue is $0$, but $\expthreshold_2$ is 10, so
the second-price auction cannot be revenue covered.  Moreover, the
welfare of this equilibrium can be far from the optimal welfare, for
example, when agent 2 has value $1$ and agent 1 has value $9$.

If agents are assumed to never bid above their values, the bidders'
threshold bids are a lower bound on the welfare of the auction, which
yields a ``welfare covering'' property which functions similarly to
revenue covering in our framework. This results in welfare guarantees
for no-overbidding equilibria that are analogous to the no-overbidding
analyses of \cite{ST13} and \cite{C14}.

\subsection{Simultaneous Composition}
\label{sec:simul}
In this section, we prove that $\mu$-revenue covering is closed under
simultaneous composition. In other words, running several
$\revpar$-revenue covered auctions simultaneously yields an auction
which is also $\revpar$-revenue covered.  Consequentially, welfare
bounds proved for individual mechanisms via revenue covering extend to
collections of such composed mechanisms.

For simplicity we assume that each component mechanism has a
``withdraw'' action which induces zero allocation and payment in that
mechansim, akin to bidding $0$ in a first-price auction. Furthermore,
we make two assumptions on agent utilities.  First, they are
\emph{unit-demand} in the sense that allocation from more than one
mechanism gives provides the same surplus as being allocated from
exactly one mechanism.  Second, they are \emph{single-valued}, meaning
agent $\agent$ has the same value $\valagent$ for allocation,
regardless of the mechanism whose allocation agent $\agent$
receives. Formally, the simultaneous composition of $m$ mechanisms for
single-dimensional agents is the following:

\begin{definition}
Let mechanisms $\mech^1,\ldots,\mech^m$ have allocation and payment
rules $(\allocitem,\paymentitem)$ for $j\in\{1,\ldots,m\}$ and
individual action spaces spaces $\actspaceagent^1,\ldots,\actspaceagent^m$ for
each agent $\agent$. The \emph{simultaneous composition} $\mech$
of $\mech^1,\ldots,\mech^m$ is defined to have:
\begin{itemize}
\item Action space $\actspaceagent=\prod_\mitem \actspaceagent^\mitem$
  for each agent. That is, each agent participates in the global
  mechanism by participating individually in each composed mechanism.
\item Vector-valued allocation rule $\left
  (\bidallocagent^1(\actions^1),\ldots,\bidallocagent^m(\actions^m)\right).$
  In other words, the mechanism gives each agent their allocated
  bundle from each mechanism. The induced single-dimensional
  allocation rule of the global mechanism is
  $\bidallocagent(\actions)=\max_{\mitem}\bidallocagentitem(\actions)$.
\item Payment rule $\bidpaymentagent(\actions)=\sum_\mitem
  \bidpaymentagentitem(\actionsitem)$. That is, agents make payments
  to every composed mechanism.
\end{itemize}
\end{definition}

Agent utilities are therefore of the form
$\bidutilagent(\actions)=\valagent
\cdot(\max_{\mitem}\bidallocagentitem(\actions))-\bidpaymentagent(\actions)
= \valagent\,\bidallocagent(\actions) - \bidpaymentagent(\actions)$.
Now using $\bidalloc$, define $\ppu$ and $\threshexpected{0}{1}$ as
discussed in \Cref{sec:frameworkpoa}. We can now state the main
theorem of the section.

\begin{theorem}
\label{thm:seq}
Let $\mech$ be the simultaneous composition of $\revpar$-revenue
covered mechanisms $\mech^1, \ldots \mech^m$ with unit-demand,
single-valued agents; then $\mech$ is $\revpar$-revenue covered.
\end{theorem}

The intuition driving the proof of \Cref{thm:seq} is that (a) the
cumulative threshold of the composite mechanism is smaller than the
cumulative threshold of each of the individual mechanisms (it is only
easier for an agent to get allocated when there are more mechanisms to
bid in) while (b) the revenue of the composite mechanism is equal to
the sum of the revenues of the individual mechanisms.

Each agent's vector of bids in the global mechanism yields a bid for
each individual mechanism. Therefore, for each individual mechanism
$\mitem$, any global bid distribution $\actionsdist$ induces a local
bid distribution $\actionsdistitem$. The distribution
$\actionsdistitem$ in mechanism $\mitem$ induces local versions of the
equivalent bid $\equivbidagentitem$, threshold $\thresholdagentitem$,
and expected threshold $\expthreshagentitem$.

We can now formalize the above intuition that individual expected
thresholds are larger than the expected thresholds in the global
mechanism:

\begin{lemma}
\label{lem:simthresh}
For any (implicit) bid distribution $\actionsdist$ in the composite
mechanism, the cumulative thresholds satisfy
$\expthresholdagent\leq\expthreshagentitem$.
\end{lemma}

\begin{proof}
Fix an individual mechanism $\mitem$. The threshold function for the
global mechanism is defined as $\ppu(\allocdev)=
\min_{\actionagent:\bidallocagent(\actionagent)\geq \allocdev}
\equivbidact{\actionagent}$, and for the local mechanism as
$\ppu^j(\allocdev)=
\min_{\actionagentitem:\bidallocagentitem(\actionagentitem)\geq
  \allocdev} \equivbidagentitem(\actionagentitem)$. For every action
$\actionagentitem$ in the local mechanism, there is a corresponding
action $\actionagent$ in the global mechanism where $\agent$ takes
action $\actionagentitem$ in mechanism $\mitem$ and withdraws from
every other mechanism. Since
$\equivbidact{\actionagent}=\equivbidagentitem(\actionagentitem)$, for
this $j$, $\actionagent$, and $\actionagentitem$ it follows that
$\ppu(\allocdev)\leq\ppu^j(\allocdev)$. This inequality only improves
if we allow any action in the composite mechanism. Integrating over
all values of $\allocdev$ implies the lemma.
\end{proof}

\begin{proof}[Proof of Theorem~\ref{thm:seq}]
An allocation $\allocalt$ is feasible in the global mechanism if and
only if there exists a feasible allocation profile $\allocaltitem$ for
each component mechanism $\mitem$ such that for each agent $\agent$,
$\allocaltagent=\sum_\mitem \allocaltagentitem$. The theorem now
follows from the following sequence of inequalities:
\begin{align*}
\revpar \rev(\mech)&=\revpar\sum\nolimits_\mitem \rev(\mechitem)\\
&\geq\sum\nolimits_\mitem\sum\nolimits_\agent \expthreshagentitem\allocaltagentitem\\
&\geq\sum\nolimits_\agent\expthresholdagent\sum\nolimits_\mitem\allocaltagentitem\\
&=\sum\nolimits_\agent \expthresholdagent\allocaltagent
\end{align*}
The first equality follows from the definition of simultaneous
composition. The second line follows revenue covering of mechanism
$\mitem$. The third is a reordering of sums and application of
\Cref{lem:simthresh}. The final inequality follows from the fact
that $\allocaltagent=\sum_\mitem \allocaltagentitem$.
\end{proof}

\Cref{thm:seq} implies that our welfare bounds extend to
simultaneous compositions of revenue-covered
mechansisms. Specifically:
\begin{corollary}
Let $\mech$ be the simultaneous composition of $m$ $\mu$-revenue
covered mechansisms for unit-demand, single-valued agents. Then the
welfare of $\mech$ is a $\mu\frac{e}{e-1}-$approximation to the
optimal welfare for the composite environment.
\end{corollary}

It is possible to extend the proof of \Cref{thm:seq} to hold for
winner-pays-bid auctions with reserves. In this case, the expected
threshold for the global mechanism with reserves $\mathbf r$ should be
defined as
$\thresholdragent{\reserveagent}=\int_{\allocmin(\reserveagent)}^1\ppu(\allocdev)\,d\allocdev$,
for $\allocmin(\reserveagent)$ defined as
$\max_{\actionagentitem:\equivbidagentitem(\actionagentitem)\leq
  \reserveagent} \bidallocagentitem(\actionagentitem).$ Here
$\allocmin(\reserveagent)$ serves the role of
$\bidallocagent(\reserveagent)$ in the winner-pays-bid definition of
$\thresholdragent{\reserveagent}$. Under the new definition of
$\thresholdragent{\reserveagent}$, one can show that revenue and value
covering with reserves extend from the single-item case to the
simultaneous composition.  Since the simultaneous composition of
first-price auctions is revenue covered with reserves, we may conclude
the following:

\begin{theorem}
Let $\mech$ be the simultaneous composition of $m$ first-price auctions with monopoly reserves $\reserves^*$, and unit-demand, single-valued agents with regular value distributions. The revenue of $\mech$ is a $\frac{2e}{e-1}-$approximation to the revenue of the optimal global mechanism. 
\end{theorem}

\section{Conclusion}
%what we showed
We have given a framework for proving worst-case approximation results
for welfare and revenue in Bayes-Nash equilibrium. This framework
enabled us to prove both welfare and new revenue approximation results
for non-truthful auctions in asymmetric settings, including first
price and all-pay auctions.

This framework has two distinct parts that isolate the analysis of
Bayes-Nash equilibrium from the analysis of the specific mechanism.
The first part, value covering, depends only on Bayes-Nash equilibrium
and relates an agent's surplus and expected threshold price to her
value. The second, revenue covering, is a property of the mechanism
which must hold for every bid distribution. This framework is
especially helpful when equilibria are hard to characterize or
understand analytically, as is the case with the first-price auction
in asymmetric environments and we expect this framework to aid broadly
in understanding properties of equilibria in auctions well beyond the
confines of classical analyses.

We invoked the characterization of Bayes-Nash equilibrium in a few
specific places in our proofs. For value covering and virtual value
covering, it is only important that an agent be best responding to the
expected actions of other bidders. For the revenue approximation
results, we do rely on the characterization of equilibrium by
\citet{M81} to account for revenue via virtual values. This crucially
allows us to relate the allocation a bidder receives to their
contribution to revenue. Extensions beyond single-parameter,
risk-neutral, private-valued agents will be challenging without a
virtual-value equivalent.

%beyond the easy case
%To extend beyond the single parameter, risk-neutral, private-values setting, a 
%notion of virtual values are needed that at least approximate the
%revenue when taken in expectation across agents. We did not rely
%on many properties of virtual values; notably, just that it
%was bounded by value and equaled revenue in expectation across
%agents. 

%\bibliographystyle{acmsmall}
\bibliographystyle{apalike}
\bibliography{bibs}

\begin{thebibliography}{}

\bibitem[Ausubel and Milgrom, 2006]{AM06}
Ausubel, L.~M. and Milgrom, P. (2006).
\newblock The lovely but lonely vickrey auction.
\newblock {\em Combinatorial auctions}, pages 17--40.

\bibitem[Borodin and Lucier, 2010]{BL10}
Borodin, A. and Lucier, B. (2010).
\newblock Price of anarchy for greedy auctions.
\newblock In {\em ACM-SIAM Symposium on Discrete Algorithms}, pages 537--553.

\bibitem[Bulow and Klemperer, 1996]{BK96}
Bulow, J. and Klemperer, P. (1996).
\newblock {Auctions Versus Negotiations}.
\newblock {\em The American Economic Review}, 86(1):180--194.

\bibitem[Caragiannis et~al., 2014]{C14}
Caragiannis, I., Kaklamanis, C., Kyropoulou, M., Lucier, B., {Paes Leme}, R.,
  and Tardos, E. (2014).
\newblock {Bounding the inefficiency of outcomes in generalized second price
  auctions}.
\newblock pages 1--45.

\bibitem[Chawla and Hartline, 2013]{CH13}
Chawla, S. and Hartline, J.~D. (2013).
\newblock Auctions with unique equilibria.
\newblock In {\em ACM Conference on Electronic Commerce}, pages 181--196.

\bibitem[Christodoulou et~al., 2013]{CKST13}
Christodoulou, G., Kov{\'a}cs, A., Sgouritsa, A., and Tang, B. (2013).
\newblock Tight bounds for the price of anarchy of simultaneous first price
  auctions.
\newblock {\em arXiv preprint arXiv:1312.2371}.

\bibitem[Dhangwatnotai et~al., 2010]{DRY10}
Dhangwatnotai, P., Roughgarden, T., and Yan, Q. (2010).
\newblock Revenue maximization with a single sample.
\newblock In {\em ACM Conference on Electronic Commerce}, pages 129--138.

\bibitem[D\"utting and Kesselheim, 2015]{DK15}
D\"utting, P. and Kesselheim, T. (2015).
\newblock Algorithms against anarchy: Understanding non-truthful mechanisms.
\newblock In {\em 16th ACM Conference on Economics and Computation}.

\bibitem[Hartline and Roughgarden, 2009]{HR09}
Hartline, J.~D. and Roughgarden, T. (2009).
\newblock Simple versus optimal mechanisms.
\newblock In {\em ACM Conference on Electronic Commerce}, pages 225--234.

\bibitem[Hoy et~al., 2015]{HNS15}
Hoy, D., Nekipelov, D., and Syrgkanis, V. (2015).
\newblock Robust data-driven efficiency guarantees in auctions.
\newblock In {\em EC Workshop on Algorithmic Game Theory and Data Science}.

\bibitem[Kaplan and Zamir, 2012]{KZ12}
Kaplan, T.~R. and Zamir, S. (2012).
\newblock Asymmetric first-price auctions with uniform distributions: analytic
  solutions to the general case.
\newblock {\em Economic Theory}, 50(2):269--302.

\bibitem[Kirkegaard, 2009]{K09}
Kirkegaard, R. (2009).
\newblock {Asymmetric first price auctions}.
\newblock {\em Journal of Economic Theory}, 144(4):1617--1635.

\bibitem[Kirkegaard, 2012]{K12a}
Kirkegaard, R. (2012).
\newblock {A Mechanism Design Approach to Ranking Asymmetric Auctions}.
\newblock {\em Econometrica}, 80(5):2349--2364.

\bibitem[Krishna, 2009]{Krishna09}
Krishna, V. (2009).
\newblock {\em Auction Theory}.
\newblock Academic Press/Elsevier.

\bibitem[Lebrun, 2006]{L06}
Lebrun, B. (2006).
\newblock {Uniqueness of the equilibrium in first-price auctions}.
\newblock {\em Games and Economic Behavior}, 55(1):131--151.

\bibitem[Lehmann and Shoham, 2002]{LOS02}
Lehmann, Daniel, I. O.~L. and Shoham, Y. (2002).
\newblock {Truth Revelation in Approximately Efficient Combinatorial Auctions}.
\newblock {\em Journal of the ACM}, 49(5):577--602.

\bibitem[Maskin and Riley, 2003]{MR03}
Maskin, E. and Riley, J. (2003).
\newblock Uniqueness of equilibrium in sealed high-bid auctions.
\newblock {\em Games and Economic Behavior}, 45(2):395 -- 409.

\bibitem[Myerson, 1981]{M81}
Myerson, R. (1981).
\newblock Optimal auction design.
\newblock {\em Mathematics of Operations Research}, 6(1):58--73.

\bibitem[Roughgarden, 2009]{R09}
Roughgarden, T. (2009).
\newblock Intrinsic robustness of the price of anarchy.
\newblock In {\em ACM Symposium on Theory of Computing}, pages 513--522.

\bibitem[Roughgarden et~al., 2012]{RTY12}
Roughgarden, T., Talgam-Cohen, I., and Yan, Q. (2012).
\newblock Supply-limiting mechanisms.
\newblock In {\em ACM Conference on Electronic Commerce}, pages 844--861.

\bibitem[Syrgkanis and Tardos, 2013]{ST13}
Syrgkanis, V. and Tardos, E. (2013).
\newblock Composable and efficient mechanisms.
\newblock In {\em ACM Symposium on Theory of Computing}, pages 211--220.

\bibitem[Vickrey, 1961]{V61}
Vickrey, W. (1961).
\newblock Counterspeculation, auctions, and competitive sealed tenders.
\newblock {\em The Journal of finance}, 16(1):8--37.

\end{thebibliography}

\appendix
\section{First-Price Welfare Approximation Lower Bound}\label{sec:app-example}

\newcommand{\actdensity}{g}
\newcommand{\utilhigh}{\util_H}

In this appendix, we describe an equilibrium of the single-item first-price auction with independent but non-identically distributed values in which the equilibrium welfare is less than that of the optimal welfare by a factor of approximately $1.15$. Our example will have $\numagents+1$ bidders. Bidders $1,\ldots, \numagents$ will be designated \emph{low-valued bidders}, with an identical value distribution to be determined shortly. Bidder $\numagents+1$ will be the \emph{high-valued bidder}, with value deterministically $1$. Allocating to bidder $\numagents+1$ in all value profiles yields a lower bound on the optimal welfare of $1$. Our constructed allocation will misallocate to low-valued bidders, yielding an expected welfare of approximately $.869$.

We will design the bid distribution of the low-valued bidders to make the high-valued bidder indifferent over an interval of bids. This will allow us to select a mixed strategy for the high-valued bidder supported on this interval. To do so, fix in advance the expected utility $\utilhigh\in[0,1]$ of the high-valued bidder. The utility $\utilhigh$ will be a parameter which  defines a family of examples constructed as below. Let $\actiondist_L$ denote the CDF of the bid distribution of an individual low-valued bidder. Then the CDF of the distribution of the highest-bidding low-valued bidder is $\actiondist_L^\numagents$. Note that if $\actiondist_L^\numagents(\action)=\utilhigh/(1-\action)$, then any bid $\action\in[0,1-\utilhigh]$ for the high-valued bidder yields an expected utility of exactly $\utilhigh$ (breaking ties in favor of bidder $\numagents+1$). We will therefore take $\actiondist_L(\action)=(\utilhigh/(1-\action))^{1/\numagents}$.

We have not yet derived a value distribution for the low-valued bidders, and we have not derived a bid distribution for the high-valued bidder. Given a bid distribution $\actiondist_H$ for the high-valued bidder, the value distribution for the low-valued bidders can be derived from first-order conditions. In other words, for any individual low-valued bidder $\agent\in 1,\ldots, \numagents$, bidder $\agent$ is facing the distribution of highest competing bid given by $\actiondist_C(\action)=\actiondist_L^{\numagents-1}(\action)\actiondist_H(\action)$. Bidder $\agent$ bids to maximize $(\valagent-\action)\actiondist_C(\action)$. For any $\actionagent\in(0,1-\utilhigh)$, first-order conditions imply that $\valagent=\actionagent+\actiondist_C(\actionagent)/\actdensity_C(\actionagent)$, where $\actdensity_C(\action)=\actiondist_C'(\action)$ is the density of bidder $\agent$'s competing bid distribution at $\action$. This mapping immediately implies a value distribution for the low-valued bidders.

All that remains is to select a mixed strategy for the high-valued bidder, their expected utility parameter, $\utilhigh$, and a number of low-valued bidders $\numagents$. To produce an equilibrium with low welfare, we must navigate a tradeoff. If $\actiondist_H$ is too aggressive, then the high-valued bidder will win frequently, yielding high welfare. If $\actiondist_H$ is too weak, then noting the formula for the low-valued bidders' values, we see that these values will generally be high. A similar tradeoff applies in selecting $\utilhigh$. Numerical experimentation shows that choosing $\actiondist_H(\action)=\sqrt{\action/(1-\utilhigh)}$ and $\utilhigh=.57$ yields low welfare. Given these choices, one can compute the expected welfare in equilibrium as approximately $.869$ for very large $n$. As mentioned, allocating the high-valued bidder yields a lower bound of $1$ on the optimal social welfare. This implies the desired approximation ratio of $1.15$.

\section{Framework Proofs}\label{sec:appframework}

\restate{Lemma~\ref{lem:vc}}{
Consider a mechanism $\mech$ in BNE with interim utility function $\utilagent(\valagent)$ for agent $\agent$ with value $\valagent$. Then
\begin{equation}
\utilagent(\valagent)+\threshexpected{\allocagent(\valagent)}{\alloclevel}\geq\tfrac{e-1}{e}\valagent.\tag{\ref{eq:vcl}}
\end{equation}}
\begin{proof}[\NOTEC{Proof }of Lemma~\ref{lem:vc}]
Note that by the definition of BNE, $\agent$ chooses an action which maximizes utility. It follows that
\begin{equation}
\label{eq:bne}
\utilagent(\valagent)\geq\valagent\allocagent(\effaction(\allocdev))-\paymentagent(\effaction(\allocdev))=\left(\valagent-\frac{\paymentagent(\effaction(\allocdev))}{\allocagent(\effaction(\allocdev))}\right)\allocagent(\effaction(\allocdev))\geq
\left(\valagent-\frac{\paymentagent(\effaction(\allocdev))}{\allocagent(\effaction(\allocdev))}\right)\allocdev.
\end{equation}

Rearranging (\ref{eq:bne}) yields 
\begin{equation}
\label{eq:cdf}\valagent-\frac{\utilagent(\valagent)}{\allocdev}\leq\frac{\paymentagent(\effaction(\allocdev))}{\allocagent(\effaction(\allocdev))}=\ppu(\allocdev).\end{equation}
%This bound is meaningful as long as $\valagent-\frac{\utilagent(\valagent)}{\allocdev}\geq 0$, or alternatively $\allocdev\geq\utilagent(\valagent)/\valagent$. It follows that
%\begin{align}
%\valagent\allocagent(\valagent) + \threshexpected{\allocagent(\valagent)}{\alloclevel}&\geq\valagent\allocagent(\valagent)+\int_{\allocagent(\valagent)}^{\alloclevel}\max\left(0,\valagent-\frac{\utilagent(\valagent)}{\allocdev}\right)\,d\allocdev\nonumber\\
%&\geq \valagent\allocagent(\valagent) + \int_{\utilagent(\valagent)/\valagent}^{\alloclevel}\valagent-\frac{\utilagent(\valagent)}{\allocdev}\,d\allocdev\nonumber\\
%&=\valagent\alloclevel+\utilagent(\valagent)\ln\frac{\utilagent(\valagent)}{\alloclevel\valagent}\label{eq:lowerbound}.
%\end{align}

%Holding $\alloclevel\valagent$ fixed and minimizing this quantity as a function of $\utilagent(\valagent)$ yields a minimum at $\utilagent(\valagent)=\frac{\alloclevel\valagent}{e}$, and at that point assumes value $(1-1/e)\alloclevel\valagent$. This is precisely the righthand side of (\ref{eq:vcl}), implying the lemma.

The remainder of the proof is identical to the proof of Lemma~\ref{lem:covering}.
\end{proof}

\section{Revenue Extension Proofs}\label{sec:app-rev}
\subsection{Auctions with Duplicate Bidders}

%For convenience, we will refer individually to the virtual surplus generated by agents with positive and negative virtual values respectively.
%\begin{definition}
%For any BNE $\strat$ of mechanism $\mech$, let $\REVpos(\mech)$ and $\REVneg(\mech)$ be the virtual surplus from agents with positive and negative virtual values, respectively. Thus, 
%\begin{align*}
%\REVpos(\mech) &= \sum_i \Ex[\valagent]{\max(0, \vvalagent(\valagent)) \allocagent(\valagent)},\\
%\REVneg(\mech) &=\sum_i \Ex[\valagent]{\min(0, \vvalagent(\valagent)) \allocagent(\valagent)}.
%\end{align*}
%\end{definition}

We first prove our extension of the analysis of \citet{BK96} to first-price and all-pay auctions with asymmetric distributions, which we restate below:

\restate{Lemma~\ref{lem:ourbk}}{
In any single-item setting with $k$-duplicates and regular value distributions, the following inequalities hold:
\begin{align*}
\REV(\FPA)\geq \tfrac{k-1}{k}\REVpos(\FPA)\\
\REV(\APA)\geq \tfrac{k-1}{k}\REVpos(\APA)
\end{align*}
}

To prove the lemma, assume $k$-duplicates holds, and consider a partition of the agents into groups $\group_1, \group_2, \ldots \group_p$ such that each group has size at least $k$ and all agents in each group $\group_i$ have values drawn from the distribution $\valuecdf$.

First note that in any BNE of a first-price or all-pay auction with $k$-duplicates, agents of the same group will play symmetric strategies. This follows from Theorem~$3.1$ of \citet{CH13}, which gives that any two agents when competing against a reserve distribution will behave identically. We obtain the following:

\begin{corollary}[of Theorem 3.1, \citet{CH13}]
In any BNE of a first-price or all-pay auctions with $k$-duplicates, for any group $\group_j$ of agents who have identically distributed values, all agents in the group play by identical strategies everywhere except on a measure zero set of values.
\end{corollary}

We now relate the revenue from each group of bidders to the revenue from a symmetric second price auction with reserves among only the bidders within the group of duplicates, allowing us to use the symmetric auction approximation results of \citet{BK96}. Let $\SPA_{R}(\group)$ be a second price auction run among agents in group $\group$ with a random reserve drawn according to the distribution $R$. 

\begin{lemma}
In any first-price or all-pay auction with $k$-duplicates, denoted $M_k$, and duplicate groups $B_1,\ldots,B_p$, there exist reserve distributions $R_1, R_2\ldots R_p$ such that 
\begin{align}
\REV(M_k) &= \sum\nolimits_j \REV( \SPA_{R_j}(\group_j)),\label{eq:sparev} \\ 
\REVpos(M_k) &= \sum\nolimits_j \REVpos( \SPA_{R_j}(\group_j)).\label{eq:sparevpos}
\end{align}
\end{lemma}
\begin{proof}
Fix the values and actions of bidders outside a group $j$ in the first-price or all-pay auction, as well as the strategies of bidders in group $j$. Fixing these results in a value $\underline v_j$ such that the highest bidder in group $j$ will get allocated if and only if their value is above $\underline v_j$. Let $R_j$ be the distribution of $\underline v_j$ induced by the random values of bidders outside group $j$. A second price auction among the members of $B_j$ with reserve price drawn from $R_j$ will induce exactly the same allocation rule as BNE of the first-price or all-pay auction for all members of the group. By revenue equivalence (Part \ref{thmpart:revequiv} of Theorem~\ref{thm:myerson}), the revenue from members of group $j$ in $M_k$ will be the same as $\REV( \SPA_{R_j}(\group_j))$. The same argument holds for $\REVpos(M_k)$ and $\REVpos( \SPA_{R_j}(\group_j))$.
\end{proof}

A second-price auction within a group is now a symmetric setting, and thus we can now use the work of \citet{BK96} to relate \eqref{eq:sparev} and \eqref{eq:sparevpos}.

\begin{proof}[Proof of Lemma~\ref{lem:ourbk}]
 By Theorem~\ref{thm:bk}, if $k\geq 2$, $\REV( \SPA_{R_j}(\group_j))\geq \frac{k-1}{k} \REVpos( \SPA_{R_j}(\group_j))$ and hence for any first-price or all-pay auction $M$:
\begin{align*}
\REV(M) &= \sum\nolimits_j \REV( \SPA_{R_j}(\group_j)) \\
			&\geq  \sum\nolimits_j \frac{k-1}{k} \REVpos( \SPA_{R_j}(\group_j)) \\
 &= \frac{k-1}{k} \REVpos(M).
\end{align*}

\end{proof}

We now prove our revenue result for the first-price auction with duplicates.

\begin{proof}[Proof of Theorem~\ref{thm:fpadupes}]
Let $\allocopt(\cdot)$ be the revenue-optimal allocation rule. For any value profile $\vals$, Lemma~\ref{lem:vvalcoveringeasy} implies:
\begin{equation}
\sum_{\agent}\vvalagent(\valagent)\allocagent(\valagent)+\sum_{\agent}\expthresholdagent\allocoptagent(\vals)\geq\frac{e-1}{e}\sum_{\agent}\valagent\allocoptagent(\vals).
\notag
\end{equation}
Taking expectations over $\vals$ and using revenue covering yields 
\begin{equation}
\frac{2e}{e-1} \REVpos(\FPA) \geq  \REV(\OPTm).
\notag
\end{equation}
By Lemma~\ref{lem:ourbk}, we may replace positive virtual values with the first price auction's expected revenue while only losing a $\tfrac{k}{k-1}$-factor. This proves the theorem.
\end{proof}

\section{Revenue Covering Proofs}\label{sec:apprevcovering}

\subsection{GFP}\label{sec:appGFP}
%In deterministic mechanisms, we used the pointwise  threshold bid $\thresholdfullagent(\valothers)$. 
%In a randomized mechanism like a position auction, fixing the actions of other agents results not in a single threshold for winning but a number of thresholds for different allocation amounts corresponding with the value of each slot. See Figure~\ref{fig:gfppointwiseT} for an illustration. As a result, the generalized first price position auction will not satisfy the same revenue covering that the normal first-price auction satisfies, but rather it will satisfy a different version which will be sufficient to prove the same approximation results.\footnote{This version of revenue covering can also be used for the other results for pay-your-bid mechanisms in Section~\ref{sec:beyondsingleitem}.} Let $T_i(\allocaltagent)$ be the expected threshold up to allocation probability $\allocaltagent$. 

In this section, we show that the generalized first-price auction satisfies revenue covering  with partial thresholds given in Definition~\ref{def:gfprevc}. For any alternate allocation $\allocalt$, the revenue will cover for each player the partial threshold up to their alternate allocation amount $\allocalt$. See Figure~\ref{fig:gfpadditionalthreshold} for an illustration.

	\begin{figure}[t]
			\small
			\centering
			%\hspace*{\fill}
			\subfloat[][ In GFP, when other players play actions $\bidothers$, there is not just one threshold for allocation, but a threshold for allocation in each slot. \label{fig:gfppointwiseT} ]{
		\begin{tikzpicture}[xscale=3.5, yscale=3.5, domain=0:0.9, smooth]
		\draw[pattern=north west lines, pattern color=lightgray] (0,0) -- (0.2, 0) -- (0.2, 0.06) -- (0.2, 0.06) -- (0.2, 0.12) -- (0.4, 0.12) -- (0.4, 0.17) -- (0.4, 0.17) -- (0.4, 0.22) -- (0.6, 0.22) -- (0.6, 0.29) -- (0.6, 0.29) -- (0.6, 0.36) -- (0.8, 0.36) -- (0.8, 0.45) -- (0.8, 0.45) -- (0.8, 0.56) -- (1, 0.56) -- (1, 0.69) -- (1, 0.69) -- (1, 0.83) -- (1.1, 0.83) -- (1.1, 0.97) -- (1.2, 0.97) -- (1.2, 1) -- (0,1) -- (0,0);
	   
\draw[-] (0,0) -- (0,1) node[left] {$1$};
\draw[-] (0,0) -- (1.4,0) node[below] {Bid ($\dev$)};
%\node at (1, -0.09) {\footnotesize $\valagent$};
%\draw[pattern=north east lines, pattern color=lightgray] ( 0.4,0) rectangle ( 1,0.22);
%\node [fill=white!50] at (0.75,0.15){\footnotesize $\bidutilagent(\bidagent)$};
\node [fill=white] at ( 0.35,0.54) {\footnotesize $ \expthreshold^{\bidothers}$};
%\draw[] (-0.02, 0.29) -- (0, 0.29);		
%\node at (0.5,-0.09) {\footnotesize $\bidagent$};
%\draw[] (1, 0) -- (1, -0.025);
%\draw[] (0.5, 0) -- (0.5, -0.025);
\node at (0.7,1.1) {Bid Allocation Rule $\ \bidallocagent(d, \bidothers)$};
\end{tikzpicture}} \IFECELSE{\hspace*{0.2cm}}{\hspace*{1cm}}
\subfloat[][For revenue covering in position auctions, the partial threshold $T^{\bidothers}_i(\allocaltagent)$ is used in place of $\allocaltagent T^{\bidothers}_i$. Note that by the convexity of $T^{\bidothers}_i(\allocaltagent)$ in $\allocaltagent$, $T^{\bidothers}_i(\allocaltagent)\leq \allocaltagent T^{\bidothers}_i.$    \label{fig:gfpadditionalthreshold}
]{
		\begin{tikzpicture}[xscale=3.5, yscale=3.5, domain=0:0.9, smooth]
		\draw (0,0) -- (0.2, 0) -- (0.2, 0.06) -- (0.2, 0.06) -- (0.2, 0.12) -- (0.4, 0.12) -- (0.4, 0.17) -- (0.4, 0.17) -- (0.4, 0.22) -- (0.6, 0.22) -- (0.6, 0.29) -- (0.6, 0.29) -- (0.6, 0.36) -- (0.8, 0.36) -- (0.8, 0.45) -- (0.8, 0.45) -- (0.8, 0.56) -- (1, 0.56) -- (1, 0.69) -- (1, 0.69) -- (1, 0.83) -- (1.1, 0.83) -- (1.1, 0.97) -- (1.2, 0.97) -- (1.2, 1);
		
		\draw[pattern=north west lines, pattern color=lightgray] (0,0) -- (0.2, 0) -- (0.2, 0.06) -- (0.2, 0.06) -- (0.2, 0.12) -- (0.4, 0.12) -- (0.4, 0.17) -- (0.4, 0.17) -- (0.4, 0.22) -- (0.6, 0.22) -- (0.6, 0.29) -- (0.6, 0.29) -- (0.6, 0.36) -- (0.8, 0.36) -- (0.8, 0.45) -- (0.8, 0.45) -- (0.8, 0.56) --(1, 0.56) -- (1, 0.83)  -- (0, 0.83) -- (0,0);
	   
\draw[-] (0,0) -- (0,1) node[left] {$1$};
\draw[-] (0,0) -- (1.4,0) node[below] {Bid ($\dev$)};
\node at (1, -0.09) {\footnotesize $\valagent$};
%\draw[pattern=north east lines, pattern color=lightgray] ( 0.4,0) rectangle ( 1,0.22);
%\node [fill=white!50] at (0.75,0.15){\footnotesize $\bidutilagent(\bidagent)$};

\node [fill=white] at ( 0.5,0.68) {\footnotesize $ \threshexpgfp{\bidothers}{0}{\allocoptagent}$};
\node at (-0.08, 0.83) {$\allocoptagent$};
\node at (-0.28, 0.22) {$\allocagent(\bidagent, \bidothers)$};
\draw[] (-0.02, 0.22) -- (0, 0.22);		
\node at (0.4,-0.09) {\footnotesize $\bidagent$};
\draw[] (1, 0) -- (1, -0.025);
\draw[] (0.4, 0) -- (0.4, -0.025);
\node at (0.7,1.1) {Bid Allocation Rule $\ \bidallocagent(d, \bidothers)$};
\draw[pattern=north east lines, pattern color=lightgray] ( 0.4,0) rectangle ( 1,0.22);
\node [fill=white!50] at (0.75,0.1){\footnotesize $\bidutilagent(\bidagent)$};
		
\end{tikzpicture}
	}	
	\caption{ \label{fig:fdfs}
	}
\end{figure}

%\begin{definition}[Revenue Covering for Position Auctions]
%Position auction $\auction$ is $\mu$-revenue covered if for any (implicit) strategy profile $\strat$, and feasible allocation $\allocalt$, 
%\begin{equation}\label{def:gfprevc}
%\revpar \REV(\mech) \geq \sum_i \threshgfpi{\allocaltagent}.
%\end{equation}
%\end{definition}

%\restate{Theorem~\ref{thm:gfp-covered}}{
	%The generalized first price (GFP) auction is $1$-revenue covered.
%}

The proof has two main steps: first we show that GFP is revenue covered with partial thresholds if bidders play deterministic bids; next, that revenue-covering with deterministic actions implies revenue covering with general distributions over bids.

\begin{proposition}\label{prop:pure}
The generalized first-price position auction is $1$-revenue covered with partial thresholds when agents play deterministic bids.
\end{proposition}
\begin{proof}

Consider the bid-based allocation rule of an agent in GFP, $\bidallocagent(\bidagent, \bidothers)$. For any bid $\bidagent$, if $\bidagent$ would be the $j$th highest bid (e.g., $j-1$ other bidders are bidding above $\bidagent$), then  $\bidallocagent(\bidagent, \bidothers)$ is the position weight of slot $j$. So, $\bidallocagent(\bidagent, \bidothers)$ is a stair function, with a stair corresponding to each position (like in Figure~\ref{fig:gfppointwiseT}). Furthermore, for a given allocation probability $z\in[0,1]$, let $\threshold^{\bidothers}(z)=\bidallocagent^{-1}(z, \bidothers)$ be the threshold price for obtaining allocation probability $z$. Note that $\threshold^{\bidothers}(\cdot)$ is also a stair function.

Let $\threshexpgfp{\bidothers}{0}{\agentallocalt} = \int_0^{\agentallocalt} \threshold^{\bidothers}(z)\ dz$ be the partial threshold when the other agents bid $\bidothers$. Denote by $b^j$ and $b^j_{-i}$ the $j$th highest bid from all bidders including and excluding $i$, respectively; then
\begin{equation}\label{eq:gfpvcg}
\threshexpgfp{\bidothers}{0}{\alpha_j} = \sum_{i=j}^{m} (\alpha_i-\alpha_{i+1}) b^j_{-i}.
\end{equation}

The revenue in GFP is $\REV(\mech(\bids))=\sum_j \alpha_j b^j \geq \sum_j \alpha_j b^j_{-i}$. For any slot $j$, the partial threshold amount for the bidder allocated $j$ in the alternate allocation is less than payment of the bidder who won the slot $j$: $\alpha_j b^j \geq \sum_{i=j}^{m} (\alpha_i-\alpha_{i+1}) b^j_{-i}$. Summing over all bidders gives $\REV(\mech(\bids))\geq \sum_i \threshexpgfp{\bidothers}{\bidallocagent(\bidagent, \bidothers)}{\alloclevel_i}$, our desired result.
\end{proof}

\newcommand{\oldthreshold}[2]{\expthresholdagent(#2)}

We now show that GFP is revenue-covered with general distributions over bids, by showing that partial thresholds with distributions over bids are less than the expectation of the partial thresholds with deterministic actions: 
\begin{equation}
\threshexpgfp{}{0}{\allocaltagent} \leq \Ex[\actionothers]{\threshexpgfp{\bidothers}{0}{\allocaltagent}}.\label{eq:desired}
\end{equation}

\begin{proof}[Proof of Lemma~\ref{thm:gfp-covered}]
Fix a distribution over bids $\actionsdist$. Consider the problem of a bidder who can react to the bids of other bidders, and wants to bid so as to minimize her expected threshold while getting allocation at least $\allocaltagent$. She solves the following optimization problem:

\begin{align}
\min_{\pi(\cdot)}\ & \Ex[\actionothers]{\threshexpgfp{\bidothers}{0}{\allocagent(\pi(\bidothers), \bidothers)}}\label{eq:minproblem} \\ 
&\text{s.t.}\ \Ex[\actionothers]{\allocagent(\pi(\bidothers), \bidothers)}\geq \allocaltagent\nonumber
\end{align}
One strategy is to always bid to get allocation $\allocaltagent$, no matter what the cost. This strategy, however,  is not optimal - the optimal strategy comes from equating marginal costs, and always buying allocation up to a price such that the average allocation purchased is at least $\allocaltagent$. That price is exactly $\bidfi(\allocaltagent)$. Thus, always bidding $\bidfi(\allocaltagent)$ is exactly the expected threshold minimizing strategy that solves \eqref{eq:minproblem}, and hence comparing with always bidding to get allocation $\allocaltagent$ gives
\begin{equation} \label{eq:gfpthings}
\Ex[\vals]{\threshexpgfp{\bidothers}{0}{\bidallocagent(\bidfi(\allocaltagent), \bidothers)}} \leq \Ex[\valothers]{\threshexpgfp{\bidothers}{0}{\allocaltagent}}.
\end{equation}

Equation~\eqref{eq:gfpthings} now allows us to derive general revenue covering from revenue covering with deterministic bids, Proposition~\ref{prop:pure}.
Noting that the left side of Equation~\eqref{eq:gfpthings} is $\threshexpgfp{}{0}{\allocaltagent}$ and taking expectation over partial thresholds with deterministic actions gives our desired result,
\begin{align}
\threshexpgfp{}{0}{\allocaltagent} &= \Ex[\actionothers]{\threshexpgfp{\bidothers}{0}{\bidallocagent(\bidfi(\allocaltagent), \bidothers)}}\\ 
 &\leq \Ex[\actionothers]{\threshexpgfp{\bidothers}{0}{\allocaltagent}} \\
&\leq \Ex[\actionothers]{\REV(\mech, \actionsdist)}\\
&= \REV(\mech,\actionsdist)\qedhere.
\end{align}
\end{proof}

\section{All-Pay Mechanisms}\label{sec:allpayapp}

The proof of \Cref{lem:aprc} lost a factor of $2$ translating all-pay
bids into their first-price equivalents. By moving from a first-price
centric version of the value covering and revenue covering framework
to one which works directly in terms of all-pay bids, we can match the
welfare results of \citet{ST13}. We can also derive a tighter revenue
result with duplicates than was possible in the first-price-based
framework.

The key quantity in our framework was the expected first-price
threshold bid. The key step to deriving an all-pay native version of
our framework is to switch to expected all-pay threshold bids. To this
end, let $\thresholdbidagent(\allocdev)$ be the inverse of the CDF of
agent $\agent$'s all-pay threshold bid. Formally,
$\thresholdbidagent(\allocdev)=\min\{b\,|\,\bidallocagent(b)\geq
\allocdev\}$. As in the first-price auction, we can compute the
expected value of this threshold bid as $\expthresholdbidagent =
\int_0^1 \thresholdbidagent(\allocdev)\,d\allocdev$. For the
first-price auction, we used the payment semantics to derive a
distribution of threshold bids for which $\agent$ would be indifferent
between all bids less than $\valagent$. We can do the same thing for
the all-pay auction and get the following result:

\begin{lemma}[All-Pay Value Covering]
\label{lem:apvc}
For any BNE of an all-pay auction and agent $\agent$ with value $\valagent$,
\begin{equation}
\utilagent(\valagent)+\expthresholdbidagent \geq \frac{\valagent}{2}.
\notag
\end{equation}
\end{lemma}
\begin{proof}
The proof parallels that of Lemma~\ref{lem:covering} - we lower bound $\expthresholdbidagent$ using the payment semantics of the all-pay auction, then minimize the lower bound. As with Lemma~\ref{lem:covering}, the deviation-based approach of \cite{ST13} also suffices.

 \textbf{Lowerbounding $\expthresholdbidagent$.} Bidder $\agent$ chooses a best response bid $\bidagent$ which maximizes her utility, $\bidutilagent(\bidagent)=\valagent\bidallocagent(\bidagent)-\bidagent$. It follows that for any other deviation
bid $\dev$,  $\bidutilagent(\bidagent)\geq\valagent\bidallocagent(\dev)-\dev$. Rearranging, we get $\bidallocagent(d)\leq\tfrac{\bidutilagent(\bidagent)+\dev}{\valagent}$. Since $\bidallocagent(\dev)$ is the CDF of $\agent$'s threshold bid, we can lower bound $\expthresholdbidagent$ by integrating above the curve $\tfrac{\bidutilagent(\bidagent)+\dev}{\valagent}$. In other words: $\expthresholdbidagent\geq \int_0^1 \max (0,\valagent\allocdev-\bidutilagent(\bidagent))\,d\allocdev.$ Call the latter quantity $\threshbidbound$.

 \textbf{Optimizing $\threshbidbound$.}
Evaluating the integral for $\threshbidbound$ gives $\threshbidbound =\valagent/2-\bidutilagent(\bidagent)+\bidutilagent(\bidagent)^2/2\valagent$, hence $\utilagent+\bidutilagent(\bidagent) =\valagent/2+\bidutilagent(\bidagent)^2/2\valagent$. Holding $\valagent$ fixed and minimizing with respect to $\bidutilagent(\bidagent)$ yields a minimum at $\utilagent(\bidagent)=0$, hence
$\bidutilagent(\bidagent) + \threshbidbound \geq \valagent/2$. Using the facts that $\bidutilagent(\bidagent)=\utilagent(\valagent)$ and $\expthresholdbidagent\geq\threshbidbound$ yields the result.
\end{proof}

As in the original framework, value covering characterizes the tradeoff between an agent's utility and the difficulty they face getting allocated. Now, however, the latter quantity is represented by $\expthresholdbidagent$, which comes from all-pay rather than equivalent first-price bids. In proving revenue covering, we can therefore skip the translation from all-pay bids to equivalent first-price bids, yielding revenue covering with $\mu=1$:

\begin{lemma}[All-Pay Revenue Covering]
\label{lem:aprc_native}
For any distribution over actions $\actionsdist$ and agent $i$ in the all-pay auction, the expected revenue is at least  $\expthresholdbidagent$.
\end{lemma}

\begin{proof}
The revenue of the all-pay auction is expected sum of all bids. This is at least the expected highest bid from all agents except $i$, which is exactly $\expthresholdbidagent$.
\end{proof}

We may combine our revenue and value covering lemmas in the manner used to prove Theorem~\ref{thm:fpawelfare} to produce the welfare bound of \cite{ST13}:

\begin{theorem}
\label{thm:apnativewelf}
The welfare in any BNE of the all-pay auction is at least a 2-approximation to the welfare of the welfare optimal mechanism.
\end{theorem}

Furthermore, from Lemma~\ref{lem:apvc}, we can derive a virtual value covering result:

\begin{lemma}[All-Pay Virtual Value Covering]
For any BNE of the all-pay auction, any agent $i$ with value $\valagent$ such that $\vvalagent(\valagent)\geq 0$,
\begin{equation}
\vvalagent(\valagent)\allocagent(\valagent)+\expthresholdbidagent\geq \frac{\vvalagent(\valagent)}{2}.
\notag
\end{equation}
\end{lemma}

\begin{proof}
Since $\valagent\allocagent(\valagent)\geq\utilagent(\valagent)$, we have
$\valagent\allocagent(\valagent) + \expthresholdbidagent\geq \tfrac{e-1}{e}\valagent$. Using the fact that $\vvalagent(\valagent)\leq\valagent$ produces the desired inequality.
\end{proof}

For revenue, combining this with revenue covering and Lemma~\ref{lem:ourbk} as we did for the first-price auction in Section~\ref{sec:duplicates} yields:

\begin{theorem}
The revenue in any BNE of the all-pay auction with at least $2$ bidders from each distribution is at least a 6-aproximation to the revenue of the optimal mechanism.
\end{theorem}

Finally, note that Lemma~\ref{lem:aprc_native} (and therefore Theorem~\ref{thm:apnativewelf}) can be extended to greedy auctions in general single-parameter environments using the approach used to derive Theorem~\ref{thm:greedy}.

\end{document}